\documentclass{optica-article}

\journal{opticajournal} % for journals or Optica Open

\articletype{Research Article}

\usepackage[capitalise]{cleveref}
\usepackage{hyperref}

\usepackage[labelformat=simple]{subcaption}

\usepackage{tikz}

\usepackage[export]{adjustbox}

\tikzset{>=latex} % for LaTeX arrow head
\usepackage{ifthen}
\usepackage{xcolor}
\usepackage{siunitx}

\newtheorem{theorem}{Theorem}[section]

\usepackage{mathtools}

\usepackage{extarrows} % http://ctan.org/pkg/extarrows

\usepackage{float}

\usepackage{lineno}
\definecolor{reblack}{RGB}{0,0,0} % revision_black

\usepackage{booktabs}

\usepackage{multirow}

\begin{document}

\title{DaISy: Diffuser-aided Sub-THz Imaging System}

\author{
Shao-Hsuan Wu,\authormark{1,$\dag$} 
Yiyao Zhang,\authormark{1,2,3,$\dag$}
Ke Chen,\authormark{3,4} and 
Shang-Hua Yang\authormark{1,5,6,*}}

\address{
\authormark{1}Institute of Electronics Engineering, National Tsing Hua University, Hsinchu, 30013, Taiwan \\

\authormark{2}Department of Mathematical Sciences, University of Liverpool, Liverpool, L69 7ZL, UK \\

\authormark{3}Centre for Mathematical Imaging Techniques, University of Liverpool, Liverpool, L69 7ZL, UK \\

\authormark{4}Department of Mathematics and Statistics, University of Strathclyde, Glasgow, G1 1XH, UK \\

\authormark{5}Department of Electrical Engineering, National Tsing Hua University, Hsinchu, 30013, Taiwan \\

\authormark{6}Terahertz Optics and Photonics Center, National Tsing Hua University, Hsinchu, 30013, Taiwan \\

\authormark{$\dag$}These authors contributed equally to this work. 
}

\email{
\authormark{*}\href
{mailto:shanghua@ee.nthu.edu.tw}
{shanghua@ee.nthu.edu.tw}
} %% email address is required; see note below about the corresponding author designation

% use {asbstract*} to suppress the copyright line. Copyright information will be added in production

% The abstract should be limited to approximately 100 words. 
% It should be an explicit summary of the paper that states the problem, the methods used, and the major results and conclusions. 
\vspace{-0.3cm}
\begin{abstract}Sub-terahertz (Sub-THz) waves possess exceptional attributes, capable of penetrating non-metallic and non-polarized materials while ensuring bio-safety. 
However, their practicality in imaging is marred by the emergence of troublesome speckle artifacts, primarily due to diffraction effects caused by wavelengths comparable to object dimensions. 
In addressing this limitation, we present the Diffuser-aided sub-THz Imaging System (DaISy), which utilizes a diffuser and a focusing lens to convert coherent waves into incoherent counterparts. 
The cornerstone of our progress lies in a coherence theory-based theoretical framework, pivotal for designing and validating the THz diffuser, and systematically evaluating speckle phenomena. 
Our experimental results utilizing DaISy reveal substantial improvements in imaging quality and nearly diffraction-limited spatial resolution. 
Moreover, we demonstrate a tangible application of DaISy in the scenario of security scanning, highlighting the versatile potential of sub-THz waves in miscellaneous fields. 
\vspace{-0.4cm}
\end{abstract}

%%%%%%%%%%%%%%%%%%%%%%%%%%  body  %%%%%%%%%%%%%%%%%%%%%%%%%%
\vspace{-0.5cm}
\section{Introduction}

Terahertz (THz) region (100 GHz to 10 THz) of the electromagnetic spectrum has experienced remarkable development and found diverse applications across various fields. 
Applications of THz technology range from high data-rate communication systems~\cite{Koenig_2013_1_Wireless} to molecular spectroscopy~\cite{Kiessling_2013_2_HighPower}, remote sensing~\cite{Liu_2010_3_Broadband}, collision-availability studies~\cite{Shakeel_2021_4_Creating}, and non-destructive imaging~\cite{Kawase_2003_5_Non, Zhang_2023_CLEO_CT_VM_EE}. 
Significantly, THz waves demonstrate exceptional transmission properties through non-metallic and non-polar mediums~\cite{NAFTALY_2005_Terahertz_transmission}, rendering them highly suitable for a wide array of applications. 
Furthermore, the bio-safety profile of THz radiation, attributed to its low quantum photon energy in the meV range, further improves its suitability for practical applications~\cite{Walker_2002_10_safety_guidelines, Berry_2003_6_safety_issues, Clothier_2003_6_Effects}. 
As a result, THz imaging has become prominent in fields like biometrics, industry, and security, including the detection of illicit drugs~\cite{Federici_2005_6_THz} and non-destructive identification of electronic circuits~\cite{Park_2015_7_IC}. 
Nonetheless, the practical viability of THz imaging systems faces limitations due to the low-power nature of their THz sources and the high absorption characteristics in polarized materials. 
These factors hinder their effectiveness in various imaging applications. 
Within the broader THz region, the sub-THz range (100-300 GHz) is particularly advantageous for certain applications, utilizing high-power illumination and relatively lower absorption coefficients in polarized materials to achieve improved depth in imaging~\cite{SUN_2011_A_promising, JUNG_2012_Quantitative}. 
This advantage positions sub-THz imaging as a multifaceted tool, suitable for bio-imaging~\cite{TSENG_2015_Terahertz_Near-field, GENTE_2015_Monitoring}, material analysis~\cite{OYAMA_2009_Sub-terahertz_imaging_of_defects_in_building_blocks}, and security scanning~\cite{TZYDYNZHAPOV_2020_security}. 

Conventional sub-THz imaging typically uses raster scanning, which focuses sub-THz waves onto a minimal focal point and systematically scans the sample pixel by pixel. 
While raster scanning is capable of producing high-quality images, the cost is the prolonged data acquisition process. 
The time-consuming nature of this process is primarily attributed to the use of a two-axis motorized stage in the process. 
An alternative approach uses a galvanometer with an f-theta lens to improve imaging speed, potentially enabling real-time imaging~\cite{Yi_2021_THzImagingSystem, Wang_2022_HighSpeed_THzImaging}. 
This approach involves directing the light path to various positions using a galvanometer and realigning it into a single-pixel detector with an f-theta lens. 
However, using the galvanometer approach entails a trade-off between fields of view (FOV), resolution, and sampling rate. 
Recently, the use of room temperature array detectors for real-time sub-THz imaging, allowing simultaneous multipixel signal acquisition, is gaining popularity. 
One notable array detector is the field-effect-transistor (FET) image sensor used in sub-THz imaging cameras~\cite{Al_2012_FET_Image}. 
The FET sensor operates by exciting plasma waves in the transistor channel, generating a constant voltage across transistor junctions via nonlinear rectification~\cite{Dyakonov_1996_FET_Principle, Knap_2002_FET_Experiment_IIIV_Resonant, Knap_2002_FET_Experiment_IIIV_Nonresonant, Knap_2004_FET_Experiment_Si, Li_2023_FPA_Review}, allowing to detect THz electric field in a broad frequency range. 
However, due to the device parasitics, FET-based cameras usually operate at lower frequencies~\cite{Li_2023_FPA_Review}. 
%Moreover, using FET for imaging requires a sub-THz source as an oscillator. This suggests that the imaging quality stability might be impacted by frequency drift or oscillator instability.% 
The uncooled microbolometer array detector is one of the contenders utilizing image sensors for fast-speed sub-THz imaging. 
This sensitive detector operates on the principle of detecting thermal variations from received radiation, altering the conductivity of the thermistor material. 
Therefore, the uncooled microbolometer array detector not only ensures stable image quality but also facilitates real-time imaging. 

However, in sub-THz real-time imaging with the array detector, imaging quality is primarily degraded by speckle artifacts with bending and interference profiles. 
These unexpected artifacts result from the diffraction effect caused by the interaction between coherent sub-THz waves and imaged objects. 
To mitigate the diffraction-induced artifacts, the optical diffuser is proposed, drawing inspiration from speckle-free imaging in the visible~\cite{Ori_2014_Non-invasive} and near-infrared spectra~\cite{Ma_2018_ma2018multimode}. 
Optical diffusers primarily scatter incident electromagnetic waves across broader angles. 
This implementation disrupts wave coherence, leading to the spatial phase independence of the propagated distance, thereby mitigating the diffraction effect. 
Various optical diffusers in the THz frequency range have been explored by several research groups~\cite{Azat_2022_Ghost_imaging, Atsushi_2020_wood-plastic, Graham_2013_air-polymer, Jaax_2013_Optical}. 
For instance, 3D-printed phase diffusers made from acrylonitrile butadiene styrene (ABS) have been investigated as random masks to generate pseudo-random patterns for coherent ghost imaging~\cite{Azat_2022_Ghost_imaging}. 
Utilizing diffuser-aided random patterns to integrate randomized phases allows for image retrieval using iterative or compressed sensing algorithms. 
Other efforts assess the feasibility of using wood-plastic composites to randomize the distribution of materials with various refractive indexes to scatter the waves~\cite{Atsushi_2020_wood-plastic}. 
Another notable work is the use of air-polymer composite materials with micro-structured air bubbles for scattering the wave over $\pm 60$ degrees from the incident direction, presenting itself as a potential candidate for THz diffusers~\cite{Graham_2013_air-polymer}. 
Additionally, binary surface reliefs with sub-wavelength features arranged in a pseudo-random pattern have been fabricated as THz diffusers~\cite{Jaax_2013_Optical}, introducing varied sub-cycle delays across different positions through the diffuser. 
This approach scatters the incident wave into multiple peaks across a range of $\pm 80$ degrees. 
Nevertheless, these studies primarily focus on the scattering capability of the device, lacking a systematic approach for designing and quantifying its decoherence capability. 
Moreover, to our knowledge, these diffusers have not yet been explicitly employed in sub-THz real-time imaging, nor have they undergone comprehensive investigation regarding speckle phenomena and practical applications. 

In this work, we introduce DaISy, a Diffuser-aided sub-THz real-time Imaging System, utilizing a high-power coherent sub-THz source and an uncooled microbolometer array detector. 
DaISy incorporates a diffuser made from cost-effective and easily shaped epoxy resin (AB glue, refractive index 1.79~\cite{naftaly_2007_terahertz}), which breaks coherence and lessens speckle phenomena induced by the coherent wave. 
To evaluate the speckle phenomena, we introduce a theoretical framework that specifically assesses coherence levels, emphasizing their relation to coherence area and speckle contrast. 
The diffuser-aided implementation significantly improves imaging quality and resolution, achieving the resolution of 1.428 LP/cm (line pairs per centimeter, i.e. the unit of line pairs, which indicates the ability to resolve distinct structures within one centimeter). 
Furthermore, we showcase the practical application of DaISy in security scanning, highlighting its potential as an effective tool. 

\section{Methods for restraining speckle artifacts in coherent imaging}
\label{sec:design_theory}

In sub-THz coherent imaging, the quality of imaging is significantly affected by speckle artifacts, primarily due to notable diffraction effects. 
These artifacts manifest as speckle patterns when the inherently coherent sub-THz beam interacts with a diffractive interface. 
Additionally, the mm-level sub-THz wavelength imposes constraints on imaging resolution, limiting the ability to resolve details at a microscopic scale. 
To mitigate these issues, methods to control the coherent beam have been developed, notably through the use of a THz diffuser coupled with a focusing lens. 
This method redistributes the phase and amplitude of the wave, effectively shortening the coherence length and disrupting the wavefront uniformity. 
This, in turn, attenuates the speckle artifacts, leading to improved image quality. 
Additionally, coherence theory provides a quantitative framework for evaluating speckle phenomena. 
Through the analysis of speckle contrast associated with the coherence area, the coherence level can be deduced. 
This analysis of coherence theory also aids in determining the efficacy of the diffuser-lens configuration for speckle evaluation. 

\subsection{Coherent beam manipulation by focusing lens and designed THz diffuser}
\label{subsec:diffuser_design}

In addressing the coherence issues in sub-THz imaging, a focusing lens is adeptly paired with a designed THz diffuser to control the coherent beam. 
The focusing lens modulates the beam phase profile across varying radii, introducing distinct delays that regulate the convergence of incident waves toward a focal point. 
This convergence results in the concentration of beam intensity into a small spot, generating a radial dependence in the phase. 
Subsequently, our efforts have been focused on developing a THz diffuser capable of effectively diminishing wave coherence. 
The diffuser, characterized by surface irregularities, disperses the coherent beam, playing a crucial role in disrupting its coherence. 
This is accomplished by scattering the beam at various angles, introducing a random phase distribution across different angles and thereby reducing its uniformity. 
This mitigation strategy aims to diminish the formation of speckle artifacts. 

To attain an effective and maximal random phase distribution, assessing phase variation (PV) $\Delta \phi$ is crucial, as the performance of the diffuser is largely dominated by the degree of PV. 
The PV measures the level to which the diffuser can alter the phase of the incoming wave, with higher values indicating greater capability for disruption. 
This capability is quantified as follows
\begin{equation}
\label{eqn_PV}
    \Delta \phi
    \, \textcolor{reblack}{=} \,
    \frac{2 \pi (n - n_{0})}{\lambda \, \textcolor{reblack}{N_{V}}}
    \sum_{i = 1}^{\textcolor{reblack}{N_{V}}}
    \left(
    \textcolor{reblack}{H_{v_{i}}} - \bar{H}
    \right)^{2}
\end{equation}
where 
$n - n_{0}$ is the change in the refractive index of the medium (in our case, $n$ for the diffuser material, $n_{0} \approx 1$ for the air), 
$\lambda$ is the wavelength of the incident beam passing through the diffuser, 
\textcolor{reblack}{
$N_{V}$ is the total number of scanning vertices of the diffuser surface, 
$V = \{ v_{i} \}$ defines the set comprising all such vertices for $i = 1, \ldots, N_{V}$, 
and $\bar{H}$ is the mean of all heights across these vertices, with 
$H_{v_{i}}$ indicating the height of each individual vertex.} 

In accordance with~\cref{eqn_PV}, the optimal design of the diffuser is characterized by a high refractive index, ensuring maximum efficacy in speckle reduction, which is pivotal for advancing the quality and resolution of coherent sub-THz imaging. 
Epoxy resin (AB glue), with its high refractive index ($n \approx$ 1.79~\cite{naftaly_2007_terahertz}), cost-effectiveness, malleability, and ease of manufacture, is a prime candidate for prototyping. 
The high refractive index of AB glue plays a crucial role in enabling more randomized PV, a key element in our diffuser design. 
Moreover, the potential of such materials can be improved by altering their intrinsic or spatial properties through additives. 

\pagebreak

The fabrication of AB glue diffusers requires a gradual solidification process, typically taking several hours after mixing. 
This slow solidification tends to evenly smooth the shape of the AB glue, challenging in achieving randomized PV. 
Experimentally, we fabricate five sets of AB-glue diffusers and shape them with varying waiting periods after mixing, ranging from 150 to 255 minutes with a 15-minute waiting interval in between (refer to~\cref{p_Img_1_8} for sample illustrations), to reduce potential errors arising from manual procedures. 
Moreover, a consistent fabrication temperature of 25 $^\circ$C is maintained for all diffusers due to the sensitivity of the exothermic nature in the solidification process with AB glue to the ambient temperature \textcolor{reblack}{(see~\hyperref[supp]{Supplement 1})}. 

For the evaluation of the PV of each diffuser after fabrication, the blue-light 3D Scanner (AutoScan Inspec, SHINNING 3D) is used to capture the height information of the diffuser surface vertices with an accuracy below 10 $\mathrm{\mu}$m. 
Then, the PV of each diffuser is quantified, as shown in~\cref{p20230821_Diffuser_PhaseVar_CohArea_SpeckleCon_1}, indicating the upward trend. 
This upward trend in PV corresponds to increased surface roughness, resulting from higher viscosity due to longer elapsed times before shaping. 
The increased surface roughness breaks the uniformity of coherent waves by scattering them, which in turn introduces PV, randomizes wavefront, and reduces coherent interference. 
\textcolor{reblack}{
Note that a saturation point exists in PV, beyond which further waiting does not lead to considerable improvements. 
This saturation point is reached when the AB glue reaches an equilibrium in its solidification process, leading to minimal changes in surface texture. 
Despite the limitations of our prototyping AB glue diffusers, we next provide the adapted coherence theory related to PV values for speckle evaluation to assess the coherence level of the coherent beam manipulation by the THz diffuser with a focusing lens. 
}

\vspace{-0.3cm}
% Figure 1
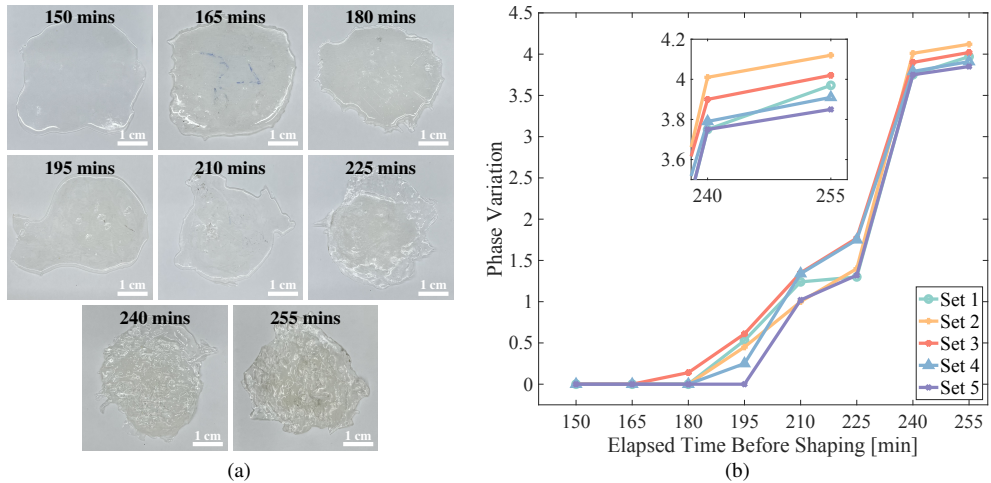
\begin{figure}[htbp]
    \centering
    \begin{subfigure}[b]{0.48\textwidth}
    \begin{tikzpicture}
    \node
    at (0, 4.5) {\includegraphics[width=.3\textwidth]{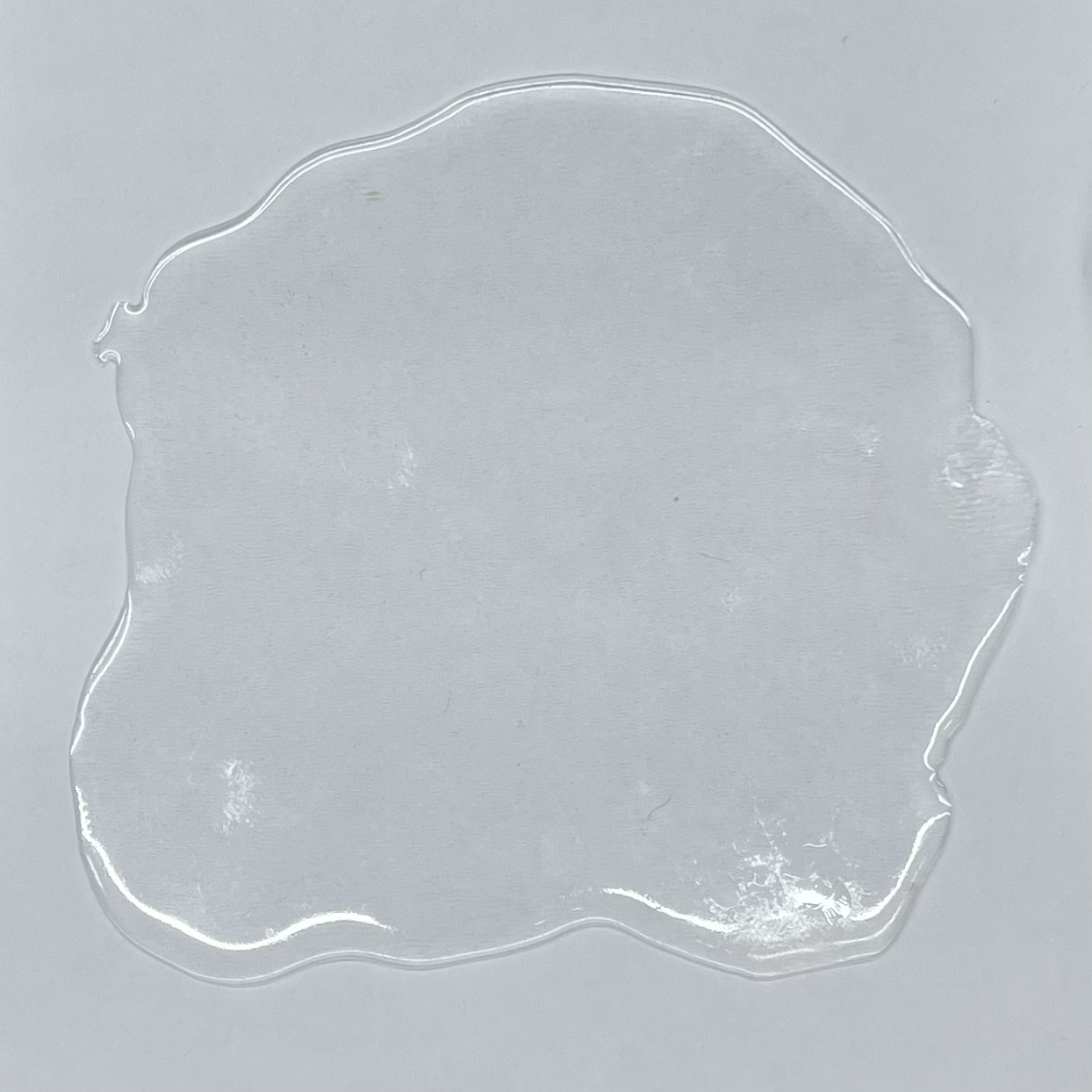}};
    
    \node[black] at (0, 5.3) {{\bf{\scriptsize{150 mins}}}};

    \draw[white,ultra thick] 
    (0.5,3.62) -- (0.9,3.62);
    \node[white] at (0.7,3.74) {\bf{\tiny{1 cm}}};

    \node
    at (2, 4.5) {\includegraphics[width=.3\textwidth]{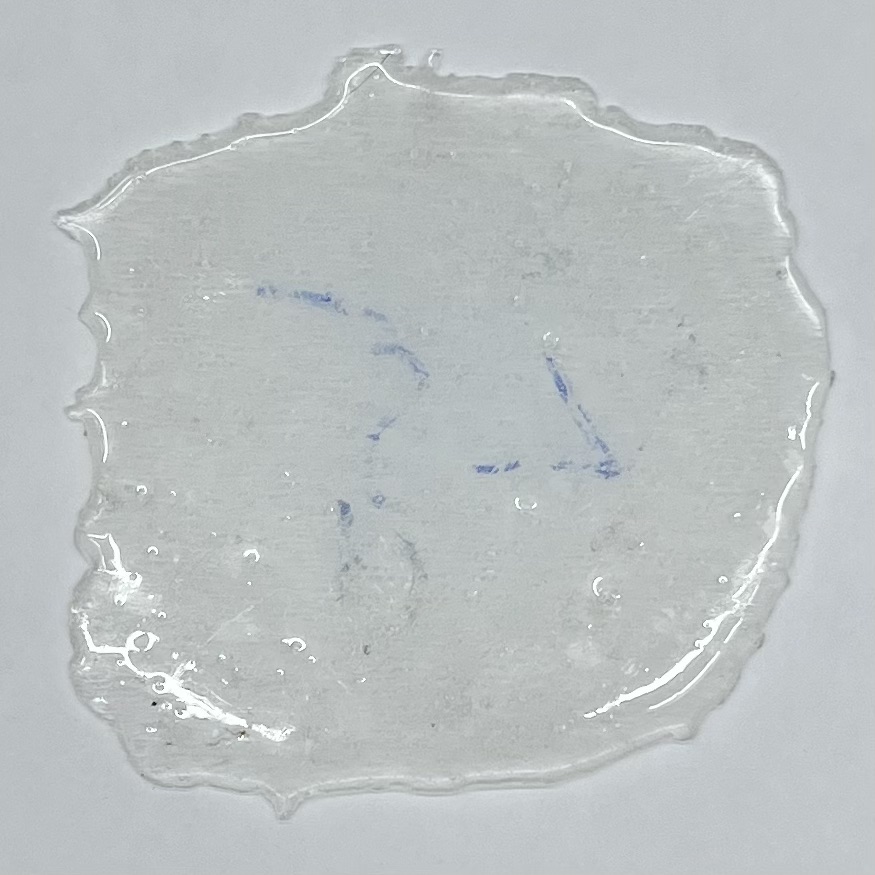}};
    
    \node[black] at (2, 5.3) {{\bf{\scriptsize{165 mins}}}};

    \draw[white,ultra thick] 
    (2.5,3.62) -- (2.9,3.62);
    \node[white] at (2.7,3.74) {\bf{\tiny{1 cm}}};

    \node at (4, 4.5) {\includegraphics[width=.3\textwidth]{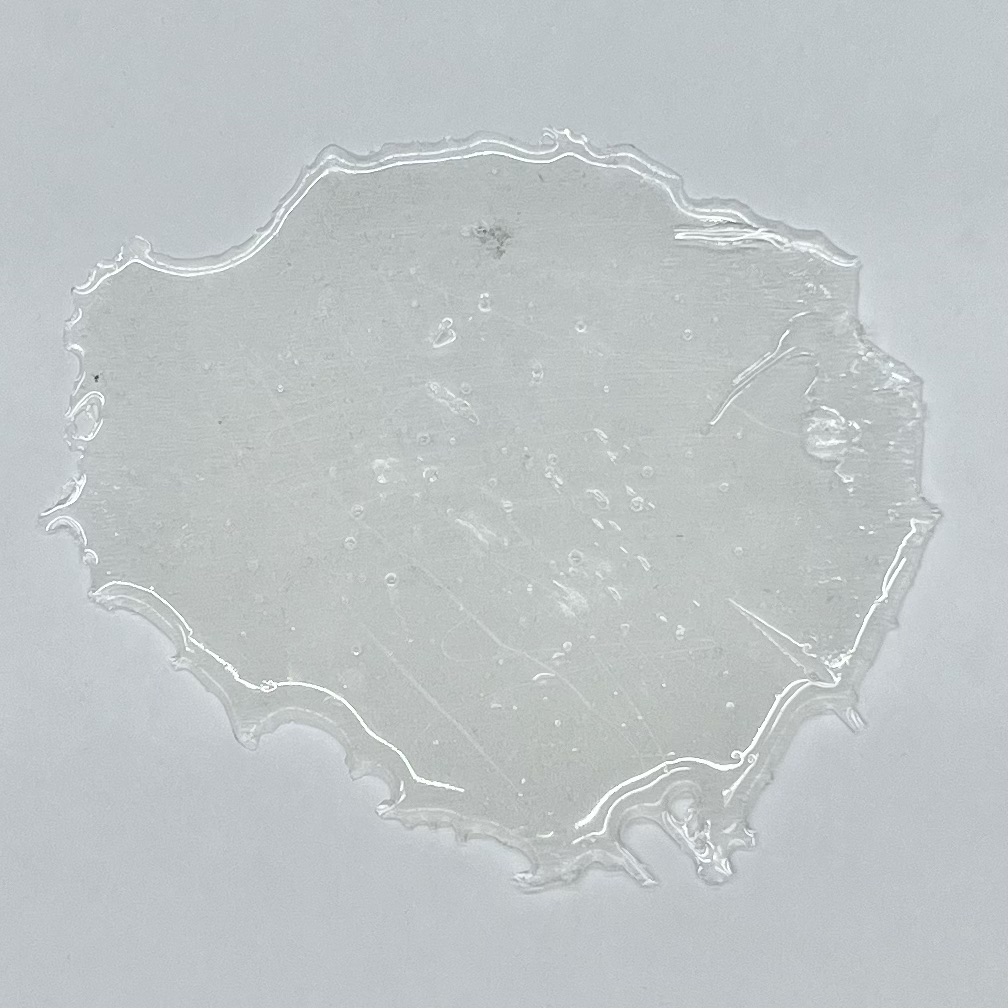}};
    
    \node[black] at (4, 5.3) {{\bf{\scriptsize{180 mins}}}};

    \draw[white,ultra thick] 
    (4.5,3.62) -- (4.9,3.62);
    \node[white] at (4.7,3.74) {\bf{\tiny{1 cm}}};

    \node
    at (0, 2.5) {\includegraphics[width=.3\textwidth]{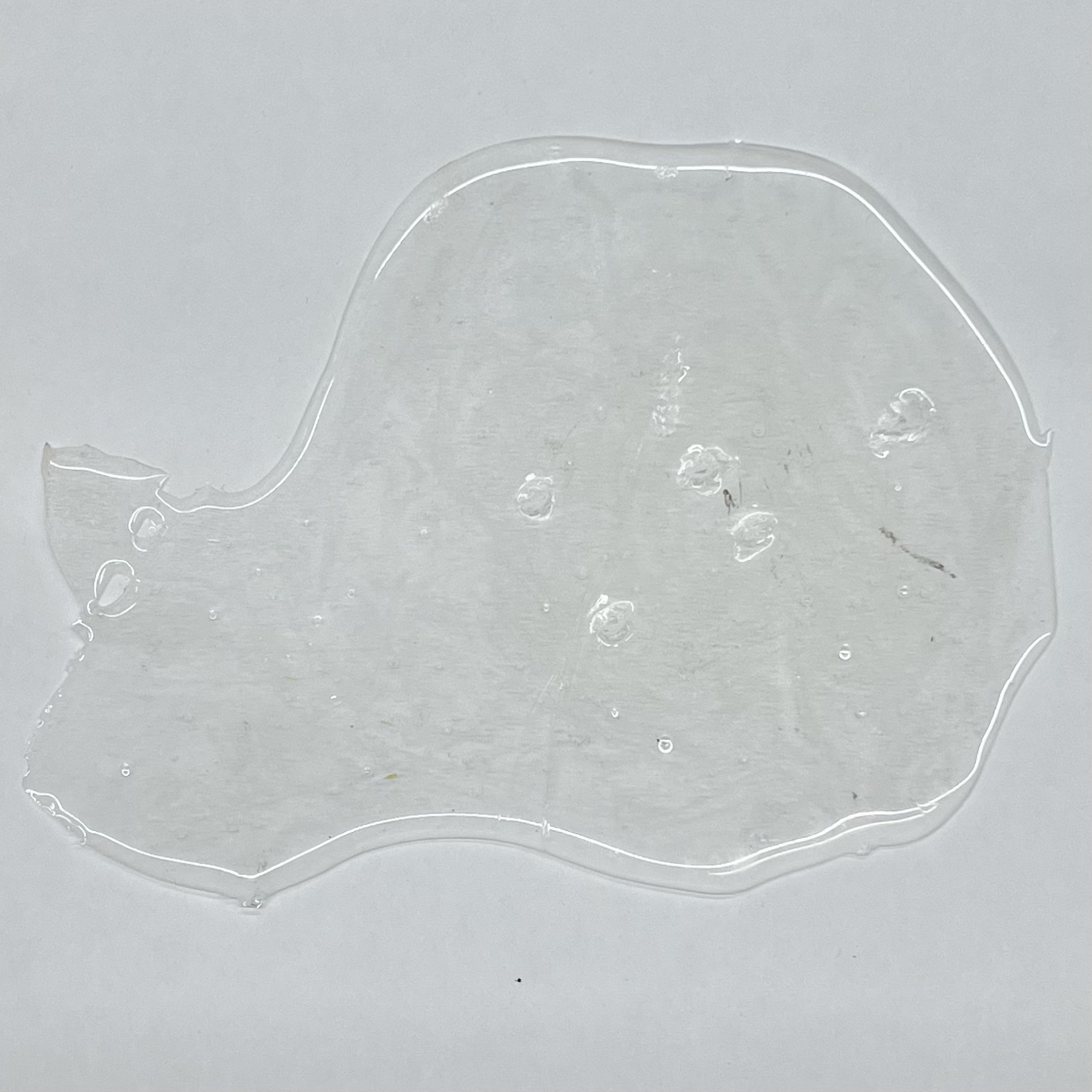}};
    
    \node[black] at (0, 3.3) {{\bf{\scriptsize{195 mins}}}};

    \draw[white,ultra thick] 
    (0.5,1.62) -- (0.9,1.62);
    \node[white] at (0.7,1.74) {\bf{\tiny{1 cm}}};

    \node
    at (2, 2.5) {\includegraphics[width=.3\textwidth]{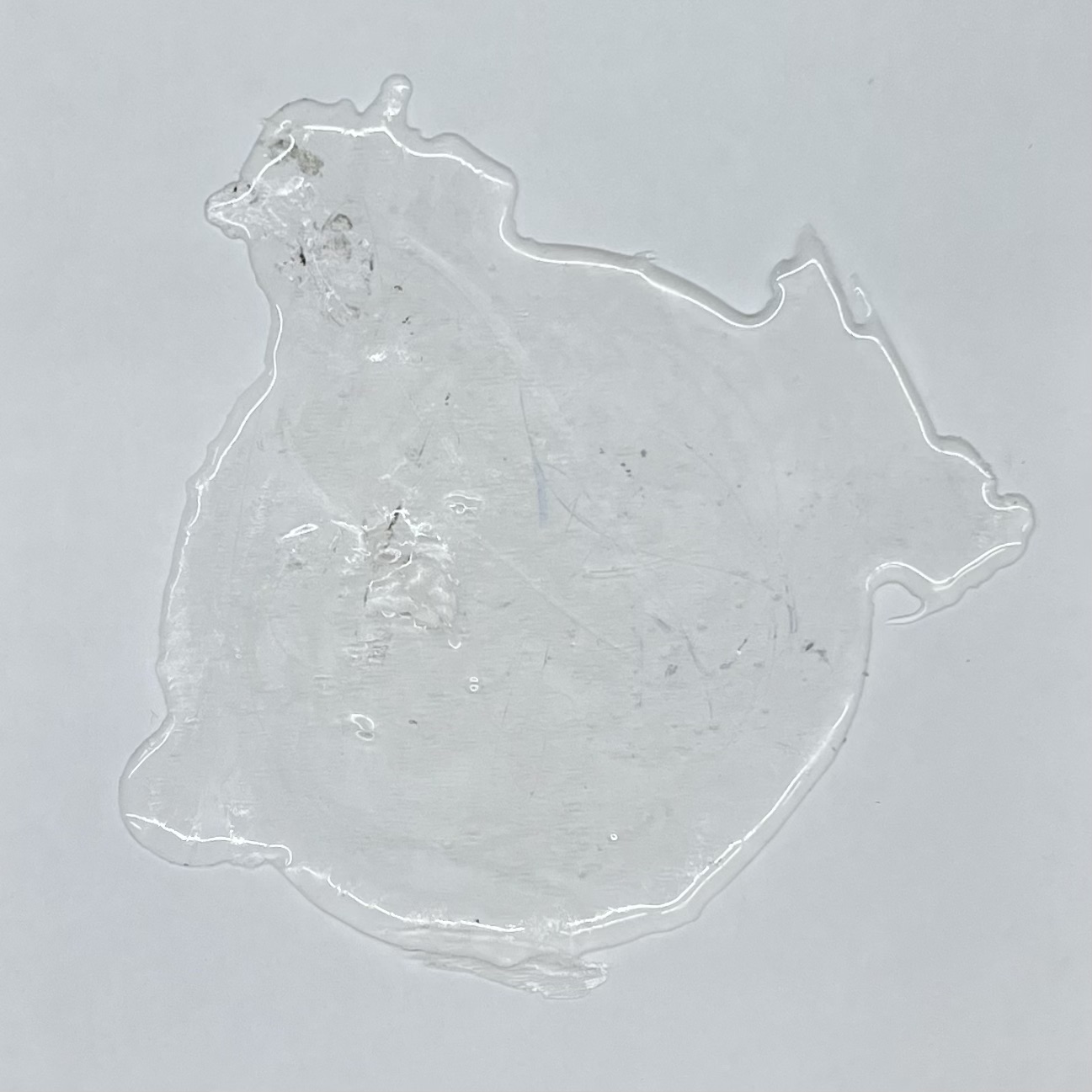}};
    
    \node[black] at (2, 3.3) {{\bf{\scriptsize{210 mins}}}};

    \draw[white,ultra thick] 
    (2.5,1.62) -- (2.9,1.62);
    \node[white] at (2.7,1.74) {\bf{\tiny{1 cm}}};

    \node
    at (4, 2.5) {\includegraphics[width=.3\textwidth]{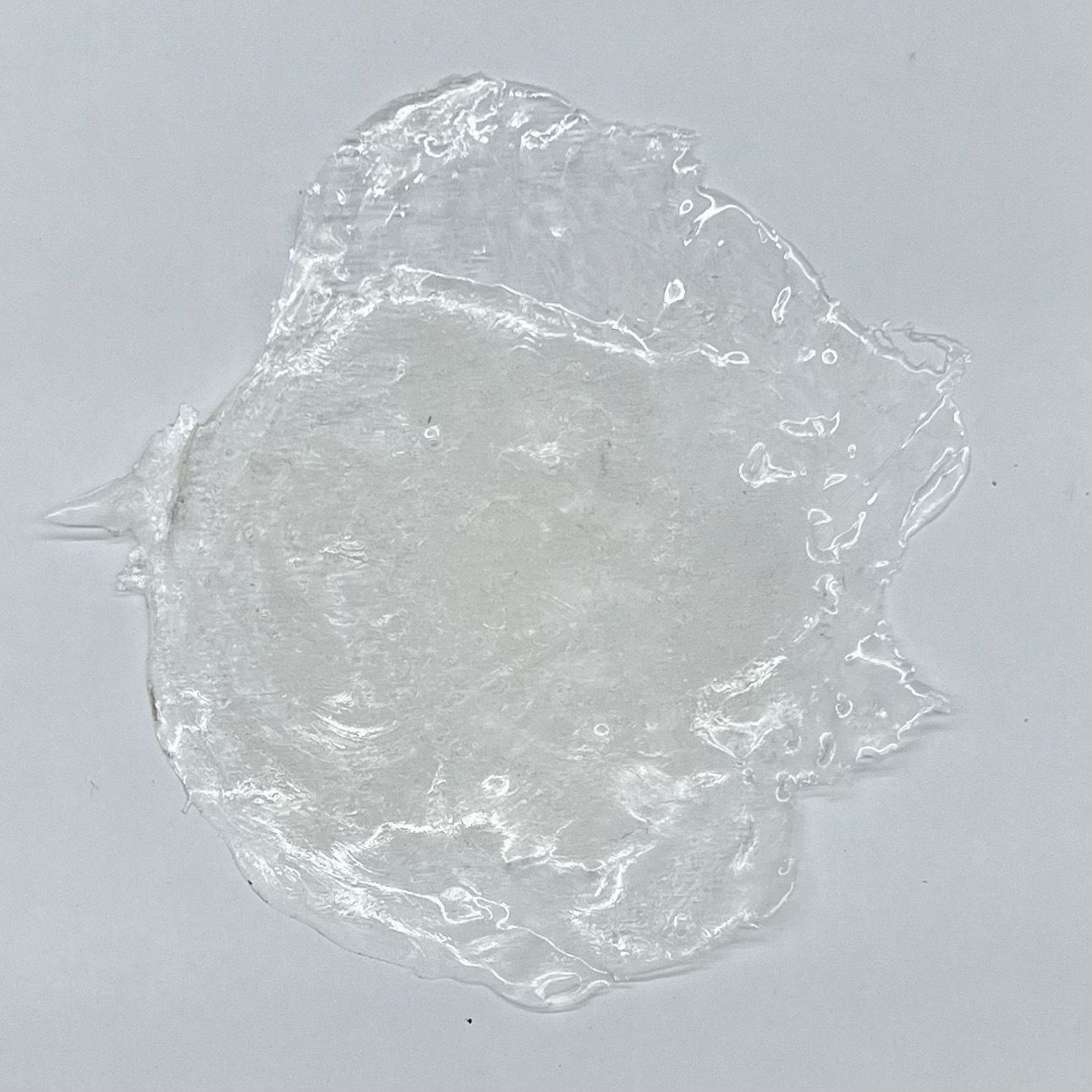}};
    
    \node[black] at (4, 3.3) {{\bf{\scriptsize{225 mins}}}};

    \draw[white,ultra thick] 
    (4.5,1.62) -- (4.9,1.62);
    \node[white] at (4.7,1.74) {\bf{\tiny{1 cm}}};

    \node
    at (1, 0.5) {\includegraphics[width=.3\textwidth]{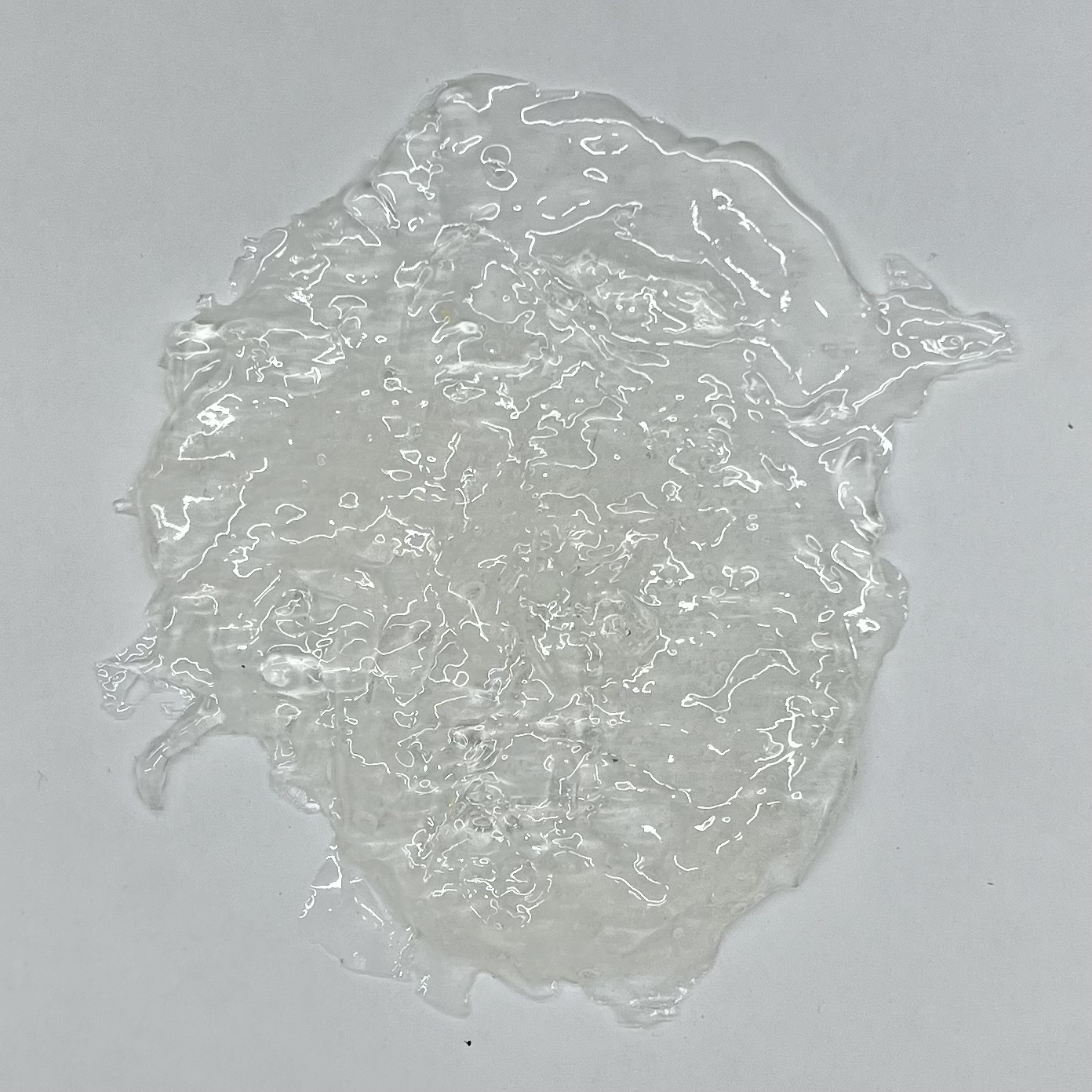}};
    
    \node[black] at (1, 1.3) {{\bf{\scriptsize{240 mins}}}};

    \draw[white,ultra thick] 
    (1.5,-0.38) -- (1.9,-0.38);
    \node[white] at (1.7,-0.26) {\bf{\tiny{1 cm}}};

    \node
    at (3, 0.5) {\includegraphics[width=.3\textwidth]{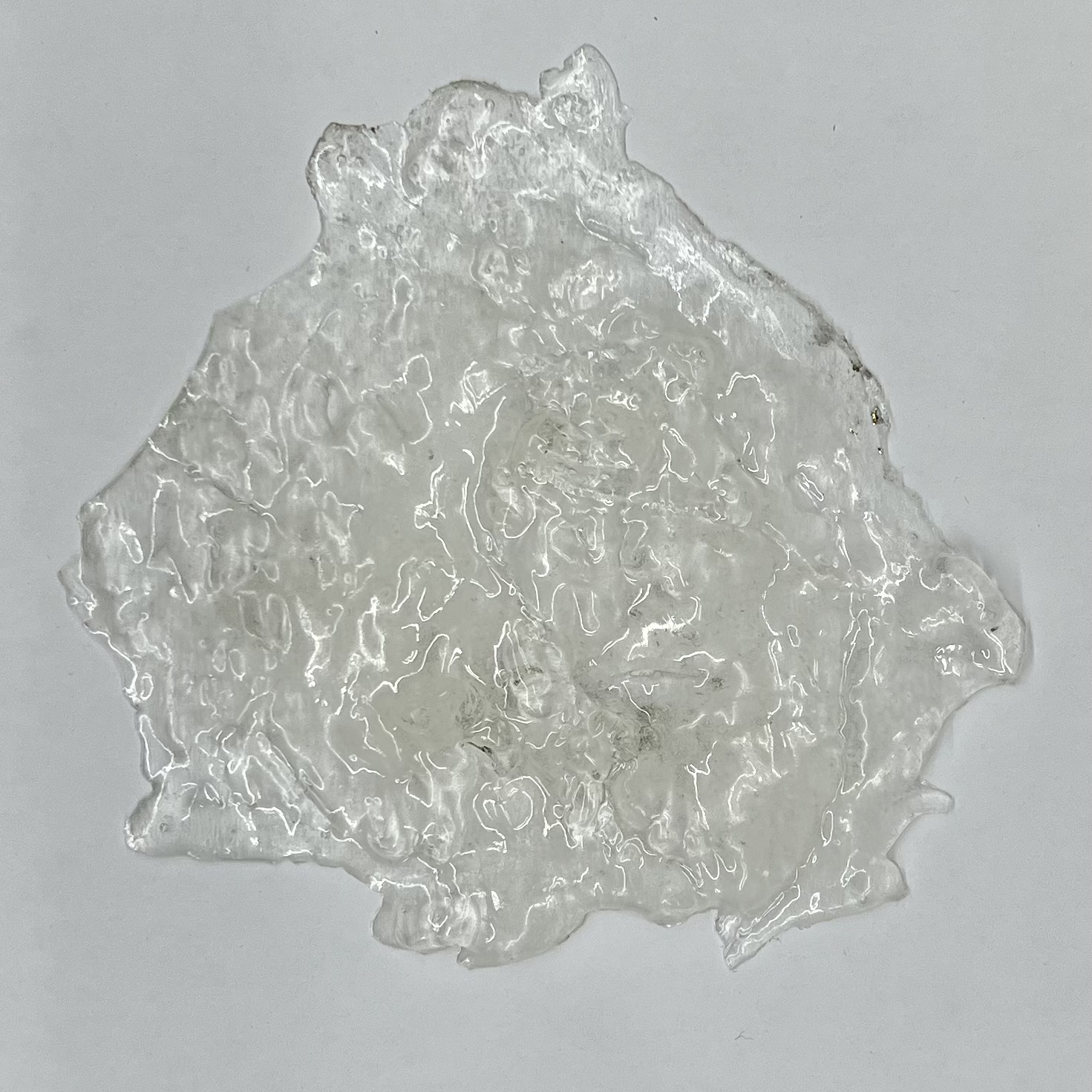}};
    
    \node[black] at (3, 1.3) {{\bf{\scriptsize{255 mins}}}};

    \draw[white,ultra thick] 
    (3.5,-0.38) -- (3.9,-0.38);
    \node[white] at (3.7,-0.26) {\bf{\tiny{1 cm}}};
    
    \end{tikzpicture}
    \vspace{-0.2cm}
    \caption{}
    \label{p_Img_1_8}
    \end{subfigure}
    \begin{subfigure}[b]{0.5\textwidth}
    \includegraphics[width=\textwidth]{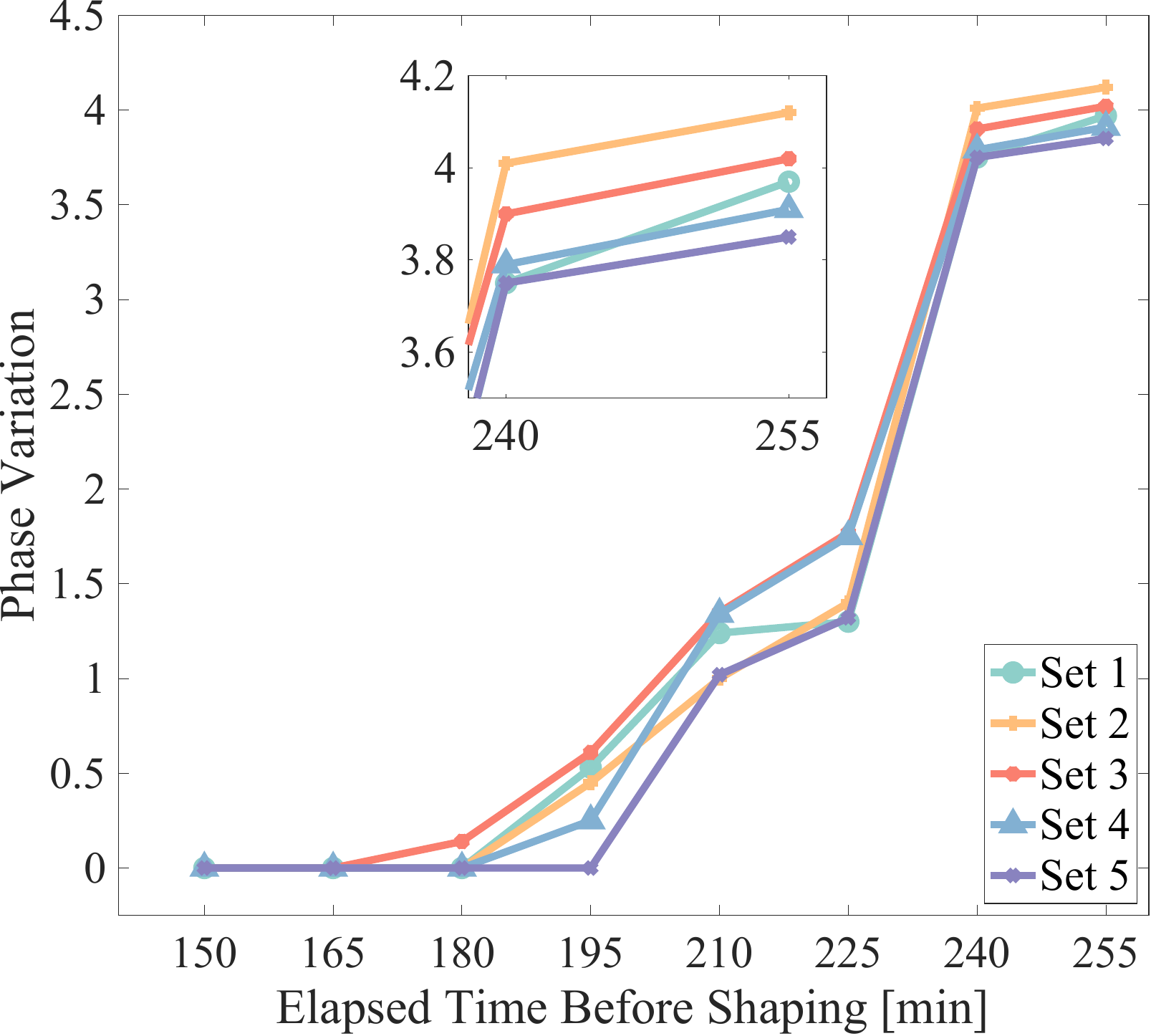}
    \vspace{-0.6cm}
    \caption{}
    \label{p20230821_Diffuser_PhaseVar_CohArea_SpeckleCon_1}
    \end{subfigure}
    \vspace{-0.3cm}
    \caption{
    Illustrations of~\subref{p_Img_1_8} AB-glue diffusers with varying waiting periods after mixing, ranging from 150 to 255 minutes with a 15-minute waiting interval in between;~\subref{p20230821_Diffuser_PhaseVar_CohArea_SpeckleCon_1} phase variations of five sets of diffusers with varying waiting times for shaping after mixing where different fabricating set of diffusers are indicated using specific colours. }
    \label{p20230821_Diffuser_PhaseVar}
\end{figure}
\vspace{-0.3cm}

\subsection{Coherence theory for speckle evaluation}
\label{subsec:theory_incorherent}

Effective mitigation of speckle artifacts is anticipated by manipulating the coherent beam through a focusing lens and a specially designed THz diffuser. 
Such manipulation of the coherent beam leads to alterations in speckle phenomena, which can be quantitatively evaluated using the concept of speckle contrast derived from optical coherence theory~\cite{GOODMAN_2007_Speckle, LI_2013_Coherence}. 
For simplicity, we first introduce coherence theory in the context of a diffuser-only configuration to derive the speckle contrast. 
Then, an adapted coherence theory, incorporating both the diffuser and the focusing lens, is presented. 
This adapted theory allows us to assess the effectiveness of these methods in reducing speckle artifacts, thereby improving image quality in coherent imaging. 

\subsubsection{Analysis with diffuser-only configuration}
\label{subsub:analysis_diffuseronly}

Beginning our analysis with a diffuser-only configuration, we explore the impact of the THz diffuser on speckle phenomena through coherence theory~\cite{GOODMAN_2007_Speckle, LI_2013_Coherence}. 
This analysis specifically evaluates the impact of the diffuser on the coherence of the beam, providing a clear basis for understanding the variations of speckle phenomena. 

Assuming a lossless THz diffuser with a stationary and steady random phase profile, we consider the scenario where the diffuser is in a rotating motion for time-varying phase alterations. 
The THz field interacting with this rotating diffuser can be mathematically expressed as
% P179 (6.2)
% P208 (6.93)
\begin{equation}
\label{eqn:ElectricField}
\mathbf{a}(\alpha, \beta; t) 
\, \textcolor{reblack}{=} \, 
\mathbf{a}_{0}
e^{j \phi_{0}(\alpha, \beta)}
e^{j \phi_{d}(\alpha - vt, \beta)}
\end{equation}
where $\mathbf{a}_{0}$ is the constant amplitude, $\phi_{0}(\alpha,\beta)$
represents the phase of incident waves at any position $(\alpha, \beta)$ in the given plane with two $\alpha, \beta$ directions, and $\phi_{d}(\alpha-vt,\beta)$ refers to the phase of the THz diffuser with a moving speed $v$ over time $t$. 
To evaluate the coherence of the THz field after passing through a diffuser, mutual coherence $\mathbf{\Gamma}_{12}$ is employed, examining the correlation between two distinct positions within an observation plane, 
% P208 (6.94)
\begin{equation}
\label{eqn_Gamma12_MutualCoherence}
\mathbf{\Gamma}_{12}(\alpha_{1}, \beta_{1}, \alpha_{2}, \beta_{2}; \Delta t) 
\, \textcolor{reblack}{=} \, 
|\mathbf{a}_{0}|^{2} 
\langle 
e^{j \phi_{d} (\alpha_{1}-vt,\beta_{1})} 
e^{- j \phi_{d} (\alpha_{2}-vt-v\Delta t,\beta_{2})} 
\rangle
\end{equation}
where $(\alpha_{1}, \beta_{1})$ and $(\alpha_{2}, \beta_{2})$ are two different positions on the same plane, $\Delta t$ represents a small time difference, and $\langle \cdot \rangle$ denotes as an infinite time average. 
Given that speckle phenomena are fundamentally determined by the phase of mutual coherence, the latter time average term of~\cref{eqn_Gamma12_MutualCoherence} crucially describes the coherence level. 
Since the diffuser follows a random phase process $\phi_{d}$, characterized by Gaussian distribution with zero mean and statistically stationary, the coherence level $\gamma_{12}$ is derived from the relationship of above average and the characteristic function (details in Thm.~\ref{thm_avg_cf} of~\Cref{appendix:derivation})
% P208 (6.95)
\begin{equation}
\label{eqn:avg_MGF}
\gamma_{12}(\Delta\alpha, \Delta\beta; \Delta t) 
\, \textcolor{reblack}{=} \, 
\langle 
e^{j\phi_{d}(\alpha_{1}-vt, \beta_{1})} 
e^{- j\phi_{d}(\alpha_{2}-vt-v\Delta t, \beta_{2})} 
\rangle 
\xlongequal{\mathrm{Thm.~\ref{thm_avg_cf}}}
e^{
-\sigma_{\phi}^{2} 
[1 - \mu_{\phi}(\Delta\alpha+v\Delta t, \Delta\beta)]
}
\end{equation}
where $\Delta\alpha = \alpha_{1} - \alpha_{2}, \Delta\beta = \beta_{1} - \beta_{2}$, $\sigma_{\phi}$ refers to the standard deviation of the diffuser phase, and $\mu_{\phi}$ represents the normalized autocorrelation function of the phase. 
Furthermore, this autocorrelation function $\mu_{\phi}$ is modeled as a Gaussian shape for subsequent analysis
% P208 (6.96)
\begin{equation}
\label{Equation:normalised autocorrelation_function_of_the_phase}
\mu_{\phi}(\Delta\alpha, \Delta\beta) 
\, \textcolor{reblack}{=} \,  
e^{-\frac{\Delta\alpha^{2} + \Delta\beta^{2}}{r_{\phi}^{2}}}
\end{equation}
where $r_{\phi}$ is the correlation radius of the phase. 
Here, the correlation radius of the phase $r_{\phi}$ indicates the degree of autocorrelation in PV caused by the diffuser. 
A higher degree of randomness in the diffuser typically results in a smaller $r_{\phi}$. 
Conversely, greater smoothness in the diffuser surface is associated with a larger $r_{\phi}$. 
Note that the slow solidification process leads to a progressive smoothing of the surface in the diffuser over the solidification period. 
This smoothing effect increases the autocorrelation value, indicating a larger $r_{\phi}$. 

Considering the transient characteristics of the observation plane, the analysis of speckle phenomena is predominantly governed by the spatial coherence level. 
In this scenario, temporal coherence is rendered negligible, i.e. $\Delta t = 0$, thereby focusing the assessment solely on spatial aspects of coherence in relation to speckle phenomena. Following~\cref{eqn:avg_MGF}-\eqref{Equation:normalised autocorrelation_function_of_the_phase}, the spatial coherence level can be reformulated as follows 
\begin{equation}
\label{Equation:gamma12_0}
\gamma_{12}(\Delta \alpha, \Delta \beta; 0) 
= 
e^{
-\sigma_{\phi}^{2}
\left(
1-e^{-\frac
{\Delta \alpha^{2} + \Delta \beta^{2}}
{r_{\phi}^{2}}
}
\right)
}. 
\end{equation}
Then, to determine the spatial coherence level over an infinite area, integration of $\gamma_{12}(\Delta \alpha, \Delta \beta; 0)$ across the variable $(\Delta \alpha, \Delta \beta)$ represents the coherence area. 
However, this direct integration may not converge due to monotonically non-decreasing integrand over an infinite interval. 
To address this issue, a renormalization process is employed, ensuring that the integral maintains unity at its origin. 
This process is achieved by adjusting the term $e^{-\sigma^{2}_{\phi}}$ from $\gamma_{12}(\Delta \alpha, \Delta \beta; 0)$. 
Specifically, the key adjustment is the division of the integral by the factor $(1-e^{-\sigma^{2}_{\phi}})$ to ensure the unity value at the origin. 
Therefore, the coherence area, following these adjustments, can be represented
% P209 (6.98)
% P210 (6.100)
\begin{equation}
\label{eqn:CoherenceArea}
\begin{aligned}
\mathcal{A}_{c} 
& = 
\iint_{-\infty}^{\infty} 
\left[
\frac
{\gamma_{12}(\Delta \alpha, \Delta \beta; 0) - \gamma_{12}(\infty, \infty; 0)}
{1 - \gamma_{12}(\infty, \infty; 0)}
\right]
\, \mathrm{d} \Delta \alpha \mathrm{d} \Delta \beta 
\\
& = 
\frac
{e^{-\sigma_{\phi}^{2}}}
{1-e^{-\sigma_{\phi}^{2}}}
\iint_{-\infty}^{\infty}
\left(
e^{
\sigma_{\phi}^{2} 
e^{-\frac{\Delta\alpha^{2}+\Delta\beta^{2}}{r_{\phi}^{2}}}
}
- 1
\right)
\, \mathrm{d}\Delta\alpha \mathrm{d}\Delta\beta. 
\end{aligned}
\end{equation}

After determining the coherence area, the speckle contrast can be derived to evaluate the speckle phenomena. 
If the diffuser is uniformly illuminated over a constant area $\mathcal{A}_{D}$, this leads to the generation of multiple coherence areas $\mathcal{A}_{c}$ on its surface. 
Each of these coherence areas $\mathcal{A}_{c}$ contributes an independent speckle intensity pattern to the observation plane. 
Then, the speckle contrast can be approximately expressed as
% (6.118)
\begin{equation}
\label{eqn:SpeckleContrast}
C 
\approx
\begin{cases}
1 & {\text{if }}\mathcal{A}_{c} \ge \mathcal{A}_{D}
\\
\left(
\frac{\mathcal{A}_{c}}{\mathcal{A}_{D}}
\right)^{\frac{1}{2}} & {\text{if }}\mathcal{A}_{c} < \mathcal{A}_{D}
\end{cases} 
\end{equation}
where $\mathcal{A}_D$ represents the constant area illuminated on the diffuser. 
With the coherence area of each diffuser in~\cref{p20230821_Diffuser_PhaseVar_CohArea_SpeckleCon_2}, the speckle contrast value can be calculated by~\cref{eqn:SpeckleContrast} for each diffuser shaped with different waiting times as shown in~\cref{p20230821_Diffuser_PhaseVar_CohArea_SpeckleCon_3}. 
As previously mentioned, the diffuser with longer waiting times exhibits a greater PV and a smaller correlation radius $r_{\phi}$ of the phase. 
Following that, these diffusers manifest a reduced coherence level, resulting in a lower speckle contrast value. 
Specifically, upon illuminating 1 m$^2$ area with a coherent beam, the minimal speckle contrast value is 0.4563, as indicated in~\cref{p20230821_Diffuser_PhaseVar_CohArea_SpeckleCon_3}. 
\textcolor{reblack}{
We choose the diffuser associated with this minimal value to break the sub-THz source coherence level externally, with the aim of achieving incoherent illuminated sub-THz imaging. 
Nevertheless, in our implementations using a single THz diffuser, the resulting imaging is still affected by a certain degree of wave diffraction. 
This limitation arises from the saturation point in PV, as discussed in~\Cref{subsec:diffuser_design}.} 

%% Figure 2
\begin{figure}[H]
    \centering
    \begin{subfigure}[b]{0.49\textwidth}
    \includegraphics[width=\textwidth]{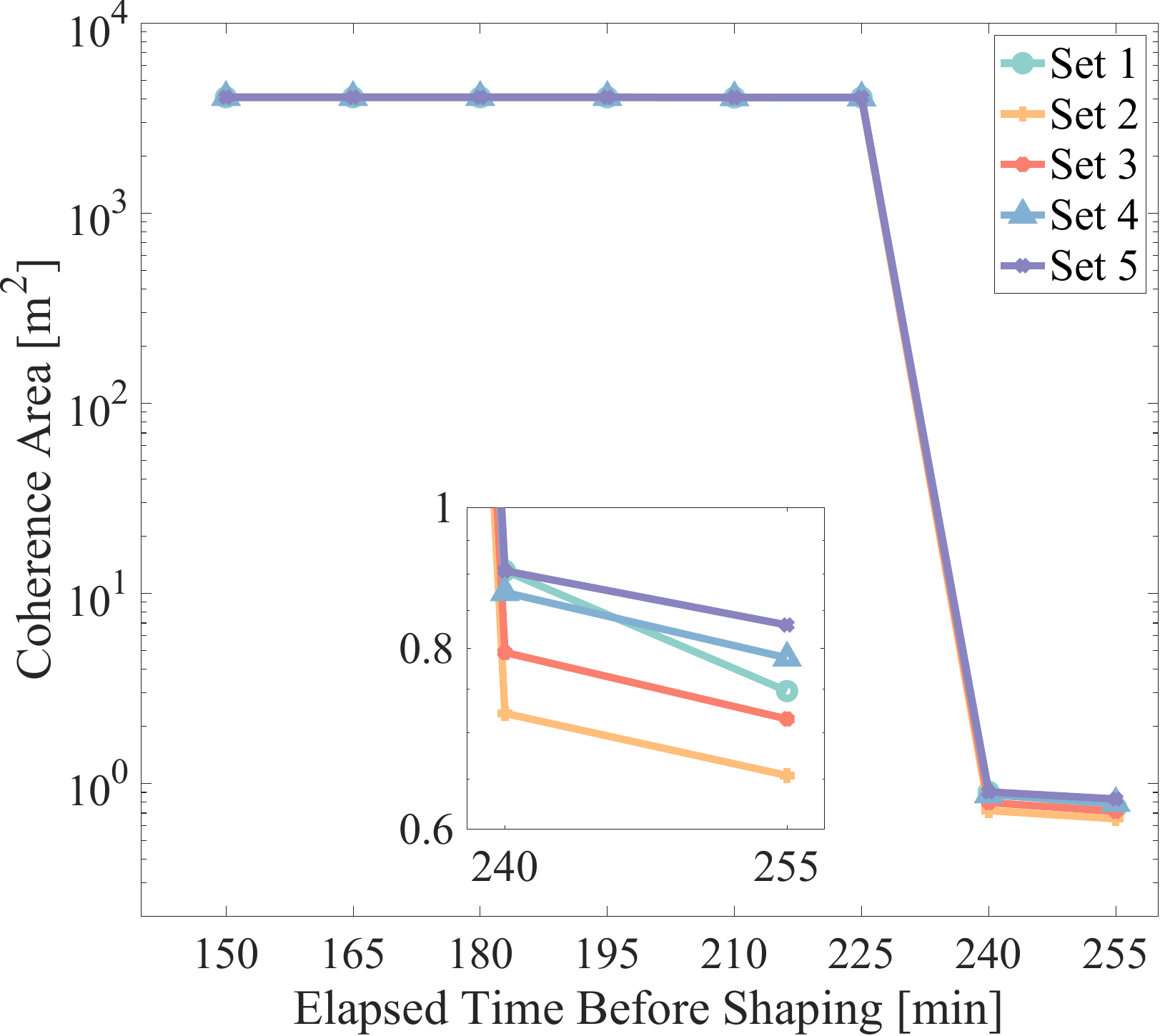}
    \vspace{-0.6cm}
    \caption{}
    \label{p20230821_Diffuser_PhaseVar_CohArea_SpeckleCon_2}
    \end{subfigure}
    \,
    \begin{subfigure}[b]{0.48\textwidth}
    \includegraphics[width=\textwidth]{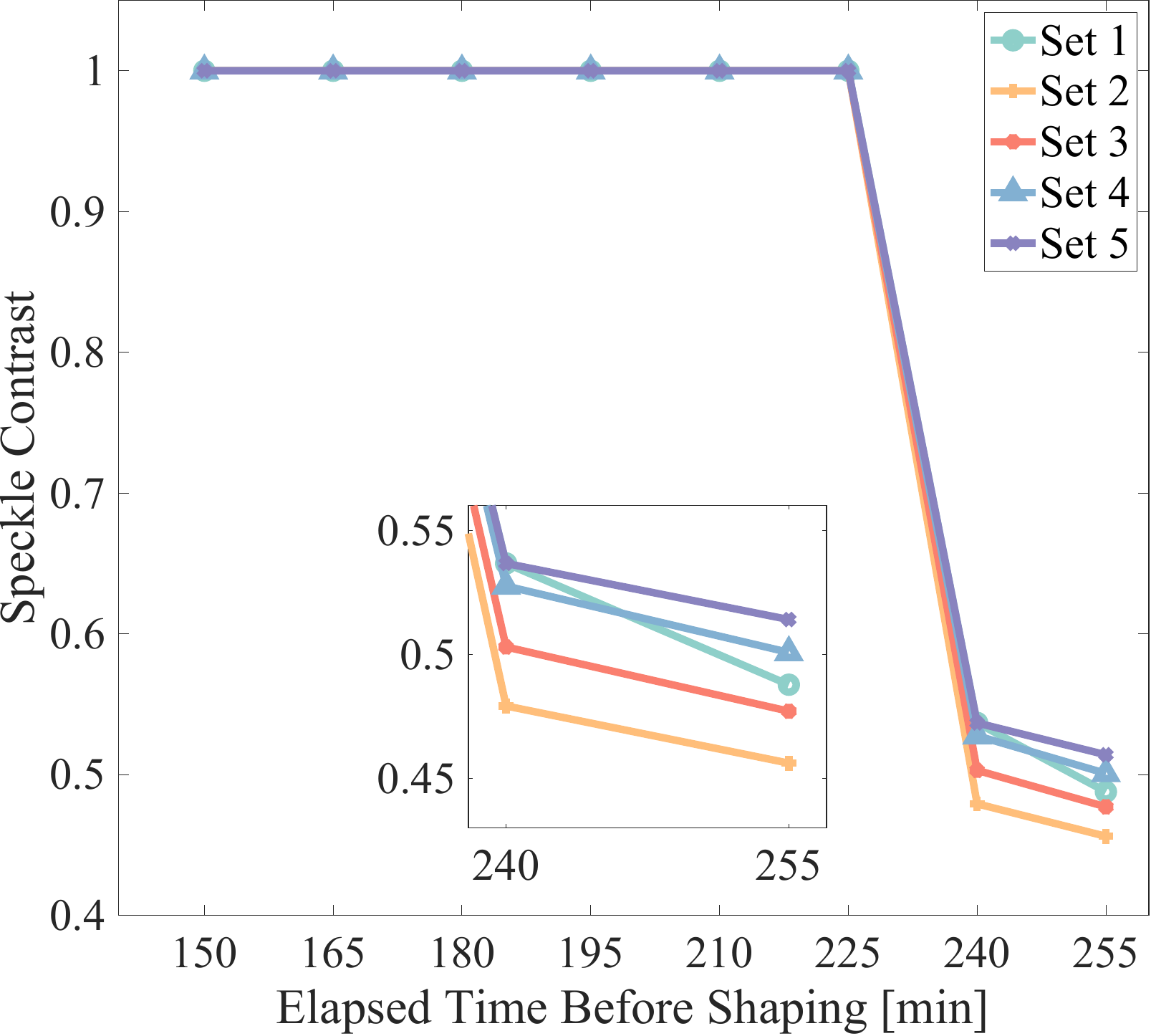}
    \vspace{-0.6cm}
    \caption{}
    \label{p20230821_Diffuser_PhaseVar_CohArea_SpeckleCon_3}
    \end{subfigure}
    \vspace{-0.3cm}
    \caption{
    Illustrations of~\subref{p20230821_Diffuser_PhaseVar_CohArea_SpeckleCon_2} coherence area of five sets of diffusers derived by~\cref{eqn:CoherenceArea};~\subref{p20230821_Diffuser_PhaseVar_CohArea_SpeckleCon_3} speckle contrast values of five sets of diffusers calculated from~\subref{p20230821_Diffuser_PhaseVar_CohArea_SpeckleCon_2} and derived by~\cref{eqn:SpeckleContrast}; where different fabricating set of diffusers are indicated using specific colours. }
    \label{p20230821_Diffuser_SpeckleCon}
\end{figure}

\subsubsection{Adapted analysis with diffuser-lens configuration}
\label{subsub:diffuserlens}

Expanding upon the initial insights from the diffuser-only configuration in~\Cref{subsub:analysis_diffuseronly}, we now explore an adapted analysis that integrates both the THz diffuser and the focusing lens within the coherence theory framework. 
In this comprehensive analysis, the interaction between the diffuser and the lens is examined to understand how their synergistic effect alters spatial coherence. 
Here, the lens not only focuses the beam but also modifies its phase characteristics, which, when combined with the random phase alterations induced by the diffuser, could result in a more pronounced reduction in speckle artifacts. 
The coherence theory is adapted accordingly to account for these manipulations, providing a more robust framework for evaluating the effectiveness of the diffuser-lens combination in speckle mitigation. 

Assuming the initial beam has a uniform phase at the same plane and passes through a thin focusing lens with a focal length $f_{l}$, the diffuser is positioned at the focus spot of this lens. 
By applying the Fresnel diffraction equation~\cite{goodman_2005_introduction}, the mathematical description of the optical field of the beam on the surface of the diffuser is as follows
% P67 (4-14); P130 (5.1)
\begin{equation}
\mathbf{U}(\alpha,\beta) 
= 
\frac{
e^{j k f_{l}}
e^{j\frac{k}{2 f_{l}}(\alpha^2 + \beta^2)}}{j\lambda f_{l}}
\iint_{-\infty}^{\infty}
\left[
\mathbf{U}_0(x,y)
e^{j\frac{k}{2 f_{l}}(x^2+y^2)}
e^{-j\frac{2\pi}{\lambda f_{l}}(\alpha x+\beta y)}
\right]
\, \mathrm{d}x\mathrm{d}y
\end{equation}
where $k$ represents the wavelength number, $f_{l}$ refers to the focal length of the thin focusing lens, $(x, y)$ represents a point on the initial plane of the light source, $\mathbf{U}_0(x, y)$ signifies the optical field at starting plane, and $(\alpha, \beta)$ represents the point on the plane situated at focal spot position $F$. 
Since $\mathbf{U}(\alpha, \beta)$ is the optical field before the incident to the diffuser,~\cref{eqn:ElectricField} can be rewritten as 
\begin{equation}
\mathbf{a}(\alpha, \beta; t) 
= 
\mathbf{U}(\alpha, \beta)
e^{j \phi_d (\alpha-vt,\beta)}. 
\end{equation}
Then, the coherence level in~\cref{eqn:avg_MGF} can be reformulated as
\begin{eqnarray}
\begin{aligned}
& \gamma_{12}(\Delta \alpha, \Delta \beta, \alpha_{2}, \beta_{2}; \Delta t) \\
= 
& \, 
e^{-\sigma_{\phi}^2[1 - \mu_{\phi}(\Delta \alpha + v\Delta t, \Delta \beta)]}
e^{j\frac{k}{2 f_{l}}[\Delta \alpha(\Delta \alpha + 2\alpha_{2}) + \Delta \beta(\Delta \beta + 2\beta_{2})]}
\iint_{-\infty}^{\infty}
e^{-j\frac{2\pi}{\lambda f_{l}}(x \Delta \alpha + y \Delta \beta)}
\, \mathrm{d}x\mathrm{d}y
\end{aligned}
\end{eqnarray}
where $\Delta\alpha=\alpha_1-\alpha_2$, $\Delta\beta=\beta_1-\beta_2$. 
Considering the spatial coherence level with $\Delta t = 0$, and $(\alpha_2,\beta_2)$ is situated at the center of the beam, the adapted coherence area in~\cref{eqn:CoherenceArea} can be derived as
\begin{equation}
\label{Eq:ac_with_phi_F}
\mathcal{\tilde{A}}_{c} 
= 
\iint_{-\infty}^{\infty}
\frac{e^{-\sigma_{\phi}^{2} + \phi_F(\Delta \alpha, \Delta \beta)}}{1-e^{-\sigma_{\phi}^{2} + \phi_F(\Delta\alpha, \Delta\beta)}}
\left\{
e^{
\left[
\sigma_{\phi}^{2} e^{-\frac{\Delta \alpha^{2} + \Delta \beta^{2}}{r_{\phi}^{2}}} 
+ 
\phi_{F}(\Delta \alpha, \Delta \beta)
\right]
} - 1
\right\}
\, \mathrm{d}\Delta \alpha \mathrm{d}\Delta \beta
\end{equation}
where $\phi_{F}(\Delta\alpha,\Delta\beta)$ refers to the contribution of coherence area from the focusing lens
\begin{equation}
\label{eqn:phi_F}
\phi_{F}(\Delta\alpha, \Delta\beta)
=
\iint_{-\infty}^{\infty}
\left[
e^{
-j\frac{k}{2 f_{l}}(\Delta \alpha^{2} + \Delta \beta^{2}) 
+ 
j\frac{2 \pi}{\lambda f_{l}}(x\Delta \alpha + y\Delta \beta)
}
\right]
\, \mathrm{d}x \mathrm{d}y. 
\end{equation}

After wave propagation through the diffuser-lens configuration, the adapted coherence area can be adjusted in various configurations associated with~\cref{Eq:ac_with_phi_F}. 
This demonstrates the potential for declining the coherence area when using a fixed phase-variation diffuser. 
Moreover, since the index of $\phi_{F}(\Delta\alpha,\Delta\beta)$ is divided by the focal length $f_{l}$, a larger value of $f_{l}$ has a more subtle effect on the coherence area, while a shorter focal length amplifies this effect. 
For example, when the maximum PV is 4.12 using a THz focusing lens ($f_{l} = 3$ cm), the speckle contrast can be reduced from 0.4563 to 0.0211. 
This diffuser-lens configuration suggests a solution to the challenge of shaping AB-glue diffusers to achieve imaging without significant distortion and a subsequent reduction in speckle artifacts. 

\pagebreak

\section{Diffuser-aided Sub-THz Imaging System (DaISy) with results and discussions}

The Diffuser-aided sub-THz Imaging System (DaISy) is configured \textcolor{reblack}{with specified distances between each component} as depicted in~\cref{p_SubTHz_Setup}. 
From right to left in the schematic diagram, the coherent sub-THz source consists of a microwave signal generator (SMB100A, Rohde \& Schwarz) and a frequency multiplier (WR-5.1, Virginia Diodes) with a conical horn antenna (WR-5.1, Virginia Diodes) optimized for wave generation and propagation. 
This assembly generates directional coherent radiation at 200.4 GHz by configuring the microwave signal generator to 16.7 GHz and amplifying it twelvefold through the frequency multiplier. 
A 2-inch plano-convex THz lens (focal length: $f_{l}$ = 3 cm) is introduced between the THz diffuser and the horn antenna, where the THz diffuser is placed at the focal point of this THz lens. 
Here, this THz lens serves to initially redistribute the phase of coherent waves, thereby amplifying the ability of the diffuser to break the coherence of incoming waves. 
\textcolor{reblack}{
All lenses employed in our DaISy, made of polypropylene (PP), are self-manufactured using computer numerical control (CNC) machinery.}
The THz diffuser is chosen with the highest PV value (PV = 4.12) to suppress the coherence level of the sub-THz radiation. 
\textcolor{reblack}{
This suppression correlates with a decrease in speckle artifacts, thereby improving the clarity and quality of the captured images.}
\textcolor{reblack}{
The THz diffuser is mounted on a rotation motor ($\mathcal{R}$) that operates at a speed of 1000 revolutions per minute (rpm).}
This rotating diffuser allows the incoming coherent wave to pass through varying thicknesses of the diffuser in a short time, effectively converting it into incoherent waves. 

As the incoherent sub-THz beam propagates through the THz diffuser, a 2-inch plano-convex THz lens ($f_{l}$ = 10 cm) is positioned to resize and efficiently collect the sub-THz beam, including the scattering components. 
Subsequently, the propagated sub-THz beam covers the whole area of the object being imaged. 
Following the object, a 4-inch plano-convex THz lens ($f_{l}$ = 17 cm) is implemented to focus the sub-THz beam onto the focal plane array of the THz camera for imaging. 
Then, the image is captured using an uncooled microbolometre array-detector THz camera (MICROXCAM-384i-THz, INO) with a THz objective lens (F/0.7, INO; $f_{l} = 4.4$ cm). % https://www.ino.ca/en/solutions/thz/lenses/
As this camera is highly sensitive to thermal variations, we attach a 0.1-mm-thick paper on top of the THz objective lens as the filter to attenuate environmental heat/light perturbations. 
All imaging results obtained from the THz camera take the normalization process from their raw files by the programming language  MATLAB\_R2023b\textsuperscript{\tiny\textregistered} performed on macOS Sonoma 14.2 with an Apple M1 Max Chip and 64 GB Memory. 

\vspace{-0.3cm}
%% Figure 3
\begin{figure}[H]
    \centering
    \begin{subfigure}[b]{0.99\textwidth}
    \adjincludegraphics[width=\textwidth, trim={{0.\width} {0.2\height} {0.\width} {0.1\height}}, clip]{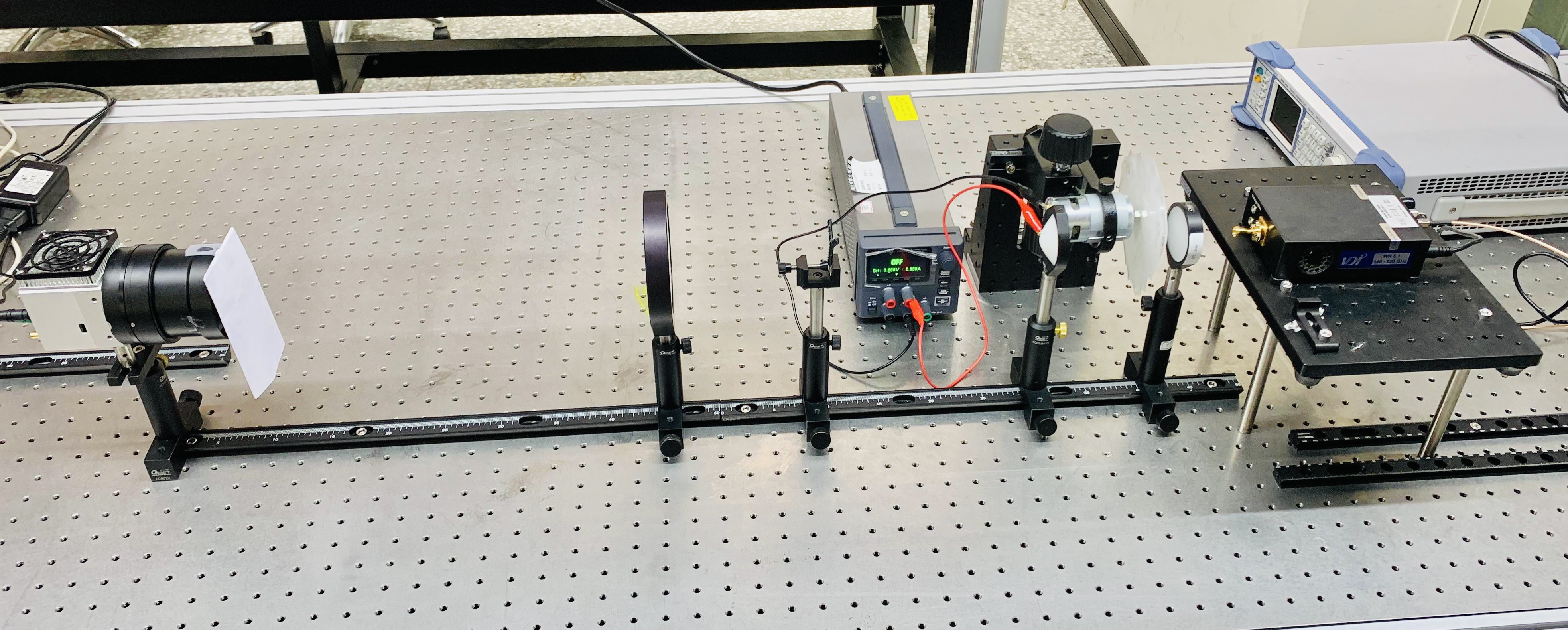}
    % \caption{}
    \label{p_Img_SubTHz_Setup}
    \end{subfigure}
    \\
    \vspace{-0.3cm}
    \begin{subfigure}[b]{\textwidth}
    \begin{adjustbox}{width=1.01\textwidth}
    \hspace{-0.6cm} 
    % \input{0_Figs/p_Sketch_SubTHz_Setup.tex}
    % 20230821_Sub-THz Setup

    \definecolor{nottblue}{RGB}{123, 192, 249} % cyan
    \definecolor{nthupurple}{RGB}{127,15,133}
    \definecolor{lightyellow}{RGB}{249, 248, 113} % lightyellow
    \definecolor{lightorange}{RGB}{255, 190, 122} % Orange
    \definecolor{goldyellow}{RGB}{255, 198, 55} % Goldyellow
    \definecolor{saffronyellow}{RGB}{255, 189, 25} % Saffronyellow
    \definecolor{ResultColour}{RGB}{191, 103, 128} 
    
    \begin{tikzpicture}[scale=1]
    \tikzstyle{every node}=[scale=1]
    
    %% Beam from Left to Right
    %% Beam 1
    \fill[ResultColour!60] (-8.3,0.9) -- (-8.3,0.3) -- (-5.9815,0.0446) -- (-5.9722,1.0361);
    %% Beam 2
    \fill[ResultColour!60] (-3.8631,0.8184) -- (-3.8661,0.3268) -- (-5.8815,0.0446) -- (-5.8722,1.0361) ;
    %% Beam 3
    \fill[ResultColour!60] (-3.2,0.6) -- (-3.8292,0.325) -- (-3.83,0.84) ;
    %% Beam 4
    \fill[lightorange!60] (-3.2,0.6) -- (-2.9,0.5) -- (-2.9,0.7) ;
    %% Beam 5
    \fill[lightorange!60] (-1.5,0.6) -- (-2.9,0.5) -- (-2.9,0.7) ;
    
    %% Objective Lens
    \fill[black!60] (-9.2,1) -- (-9.2,0.2) -- (-8.2,0.2) -- (-8.2,1) plot[smooth, tension=.7] coordinates {(-8.2,1) (-8.2,0.6) (-8.2,0.2)};
    \draw plot[smooth, tension=.7, double] coordinates {(-9.2,1) (-9.2,0.2)};
    \draw plot[smooth, tension=.7] coordinates {(-9.2,1) (-8.2,1)};
    \draw plot[smooth, tension=.7] coordinates {(-9.2,0.2) (-8.2,0.2)};
    \draw plot[smooth, tension=.7] coordinates {(-8.2,1) (-8.2,0.6) (-8.2,0.2)};
    \node at (-9.1,-0.4) {\scriptsize{with Objective Lens}};
    
    %% THz Camera
    \draw[smooth, tension=.7, fill, gray!50]  (-10.2,1.1) rectangle (-9.2,0.1);
    \node at (-9.7,-0.1) {\scriptsize{THz Camera}};
    
    %% IR Filter - Paper
    \draw[smooth, tension=.7, fill, gray!30]  (-8.18,1.15) rectangle (-8.15,0.05);
    \node at (-8.8,-0.7) {\scriptsize{and Filter (Paper)}};
    
    % Lens 1
    \fill[gray!30] plot[smooth, tension=.7] coordinates {(-6,1.4) (-5.7,0.6) (-6,-0.2) (-6,1.4)};
    \draw[white] plot[smooth, tension=.7] coordinates {(-6,1.4) (-5.7,0.6) (-6,-0.2) (-6,1.4)} ;
    \node at (-5.9,-0.4) {\scriptsize{Lens}};
    
    %% Stage
    \fill[black!80] plot  (-5.2,0.4) rectangle (-4.9,0.3);
    
    %%  Object
    \fill[gray!70] plot[smooth, tension=.7, thick] (-5,0.4) -- (-5,0.8) -- (-5.1,0.8) -- (-5.1,0.4) -- (-5,0.4) ;
    \draw plot[smooth, tension=.7, thick] (-5,0.4) -- (-5,0.8) -- (-5.1,0.8) -- (-5.1,0.4) -- (-5,0.4) ;
    \node at (-5.05,1.1) {\scriptsize{Object}};
    
    %% Lens 2
    \fill[gray!30] plot[smooth, tension=.7] coordinates {(-3.8,0.9) (-4,0.6) (-3.8,0.3)  (-3.8,0.9)};
    \draw[white] plot[smooth, tension=.7] coordinates {(-3.8,0.9) (-4,0.6) (-3.8,0.3) (-3.8,0.9)};
    \node at (-3.8,0.1) {\scriptsize{Lens}};
    % \node at (-3.8,-0.2) {\scriptsize{Lens}};
    
    %% Lens 3
    \fill[gray!30] plot[smooth, tension=.7] coordinates {(-2.9,0.9) (-2.6,0.6) (-2.86,0.3) (-2.9,0.4) (-2.9,0.5) (-2.9,0.6) (-2.9,0.9)} ;
    \draw[white] plot[smooth, tension=.7] coordinates {(-2.9,0.9) (-2.6,0.6) (-2.86,0.3) (-2.9,0.4) (-2.9,0.5) (-2.9,0.6) (-2.9,0.9)} ;
    \node at (-2.8,0.1) {\scriptsize{Lens}};
    
    %% Diffuser
    \draw[very thick, black!40]  plot[smooth, tension=.8] 
    		coordinates {(-3.184,0.2677) (-3.2606,0.0768) (-3.2314,-0.0229)};
    		
    \draw[smooth, tension=.7, fill, nthupurple!60]  (-3.2,1) rectangle (-3.15,0.2);
    \node at (-3.175,1.15) {\scriptsize{Diffuser}};
    
    \draw (-3.175,-0.1866) ellipse (0.12 and 0.18);
    \node at (-3.175,-0.1866) {\tiny{$\mathcal{R}$}};
    \node at (-2.72,-0.3481) {\scriptsize{Motor}};
    
    %% Emitter
    \fill[goldyellow] (-2.2,0.65) -- (-2.2,0.55) -- (-2.35,0.5) -- (-2.35,0.7) plot[smooth, tension=.7] 
    			coordinates {(-2.35,0.7) (-2.4,0.6) (-2.35,0.5)};
    \draw[saffronyellow] plot[smooth, tension=.7] (-2.35,0.7) -- (-2.2,0.65) -- (-2.2,0.55) -- (-2.35,0.5) ;
    \draw[saffronyellow] plot[smooth, tension=.7] coordinates {(-2.35,0.7) (-2.4,0.6) (-2.35,0.5)};
    
    %% Cable
    \draw[very thick, saffronyellow!60]  plot[smooth, tension=.8] coordinates {(-1.1,0.8) (-1.3,0.7) (-1.5,0.5)};
    
    %% Frequency Multiplier
    \fill[black!80] plot  (-1.4,0.4) rectangle (-2.2,0.8);
    \node at (-1.1,0.2) {\scriptsize{Frequency Multiplier}};
    \node at (-1.35,-0.1) {\scriptsize{with Horn Antenna}};
    % \node at (-1.85,-0.4) {\scriptsize{($\times 12$)}};
    
    %% RF Source
    \draw[fill, nottblue!60]  (0,1.4) rectangle (-1.2,0.6);
    \node at (-0.6,1.1) {\scriptsize{Microwave}};
    \node at (-0.6,0.9) {\scriptsize{Source}};
    
    % Label - Incoherence
    \draw[color=ResultColour, fill=ResultColour!60]  (-4.8,-0.6) rectangle (-4.6,-0.8);
    \node at (-3.9,-0.7) {\scriptsize{Incoherence}};
    
    % Label - Coherence
    \draw[color=lightorange, fill=lightorange!60]  (-2.9,-0.6) rectangle (-2.7,-0.8);
    \node at (-2.1,-0.7) {\scriptsize{Coherence}};
    
    \draw (-8.2,1.5)--(-8.2,1.55);
    \draw (-8.2,1.55)--(-5.9,1.55);
    \draw (-5.9,1.5)--(-5.9,1.55);
    \node at (-7.1,1.42) {\tiny{35 cm}};
    
    \draw (-5.9,1.55)--(-5.05,1.55);
    \draw (-5.05,1.5)--(-5.05,1.55);
    \node at (-5.475,1.42) {\tiny{12 cm}};
    
    \draw (-5.05,1.55)--(-3.9,1.55);
    \draw (-3.9,1.5)--(-3.9,1.55);
    \node at (-4.475,1.42) {\tiny{20 cm}};
    
    \draw (-3.175,1.55)--(-3.9,1.55);
    \draw (-3.175,1.5)--(-3.175,1.55);
    \node at (-3.5375,1.42) {\tiny{7 cm}};
    
    \draw (-3.175,1.55)--(-2.75,1.55);
    \draw (-2.75,1.5)--(-2.75,1.55);
    \node at (-2.9625,1.42) {\tiny{3 cm}};
    
    \end{tikzpicture}
    % \caption{}
    \label{p_Sketch_SubTHz_Setup}
    \end{adjustbox}
    \end{subfigure}
    \vspace{-0.5cm}
    \caption{
    Schematic diagram of the Diffuser-aided sub-THz Imaging System (DaISy). }
    \label{p_SubTHz_Setup}
\end{figure}
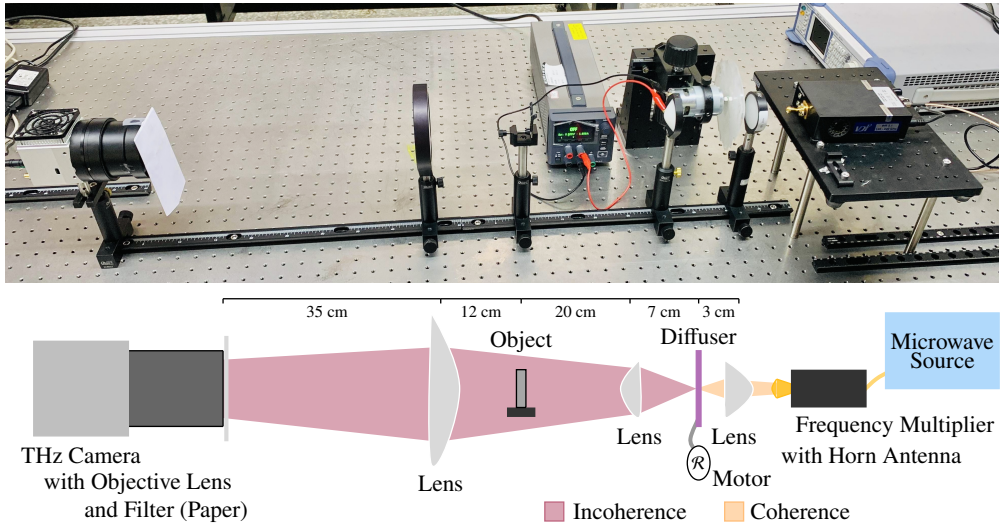

\pagebreak

As elaborated in~\Cref{sec:design_theory}, solely using a diffuser does not suffice to meet the desired DaISy performance criteria. 
To demonstrate this, experiments are conducted under three different configurations: without both a diffuser and a focusing lens, using only a diffuser, and combining both a diffuser and a focusing lens. 
Initially, we explore the formation of speckle patterns with Gaussian sub-THz beam illumination, as depicted in~\cref{p20230823_Diffuser_THzCamera_Beam}. 
In the first scenario (\cref{p20230823_Diffuser_THzCamera_Beam_1}), where neither the diffuser nor the focusing lens is used, we observe a speckle contrast of 0.409, indicating a high level of speckle artifacts due to the nature of coherence imaging. 
In the second configuration, where only the diffuser is employed (\cref{p20230823_Diffuser_THzCamera_Beam_2}), there is a noticeable decrease in speckle contrast to 0.255. 
This decrease demonstrates the effectiveness of the diffuser in altering the coherence of the beam, thereby influencing the speckle patterns. 
A further improvement is seen in the third configuration (\cref{p20230823_Diffuser_THzCamera_Beam_3}) with the addition of the focusing lens to the diffuser. 
Here, the speckle contrast is further reduced to 0.155, showcasing the improved capability of this combined configuration in reducing speckle artifacts. 
The decreasing speckle contrast values across our configurations demonstrate the effectiveness of both the diffuser and lens in breaking beam coherence, directly improving image clarity. 
This consistent trend highlights the significant role each component plays in achieving higher image quality and moving towards speckle-free imaging. 

%% Figure 4
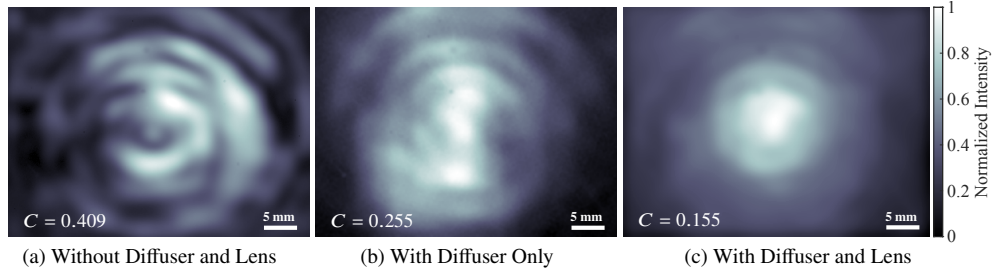
\begin{figure}[htbp]
\centering
    \begin{subfigure}[b]{0.3\textwidth}
    \begin{tikzpicture}
    \node
    {\adjincludegraphics[width=\textwidth, trim={{0.\width} {0.\height} {.18\width} {0.\height}}, clip]{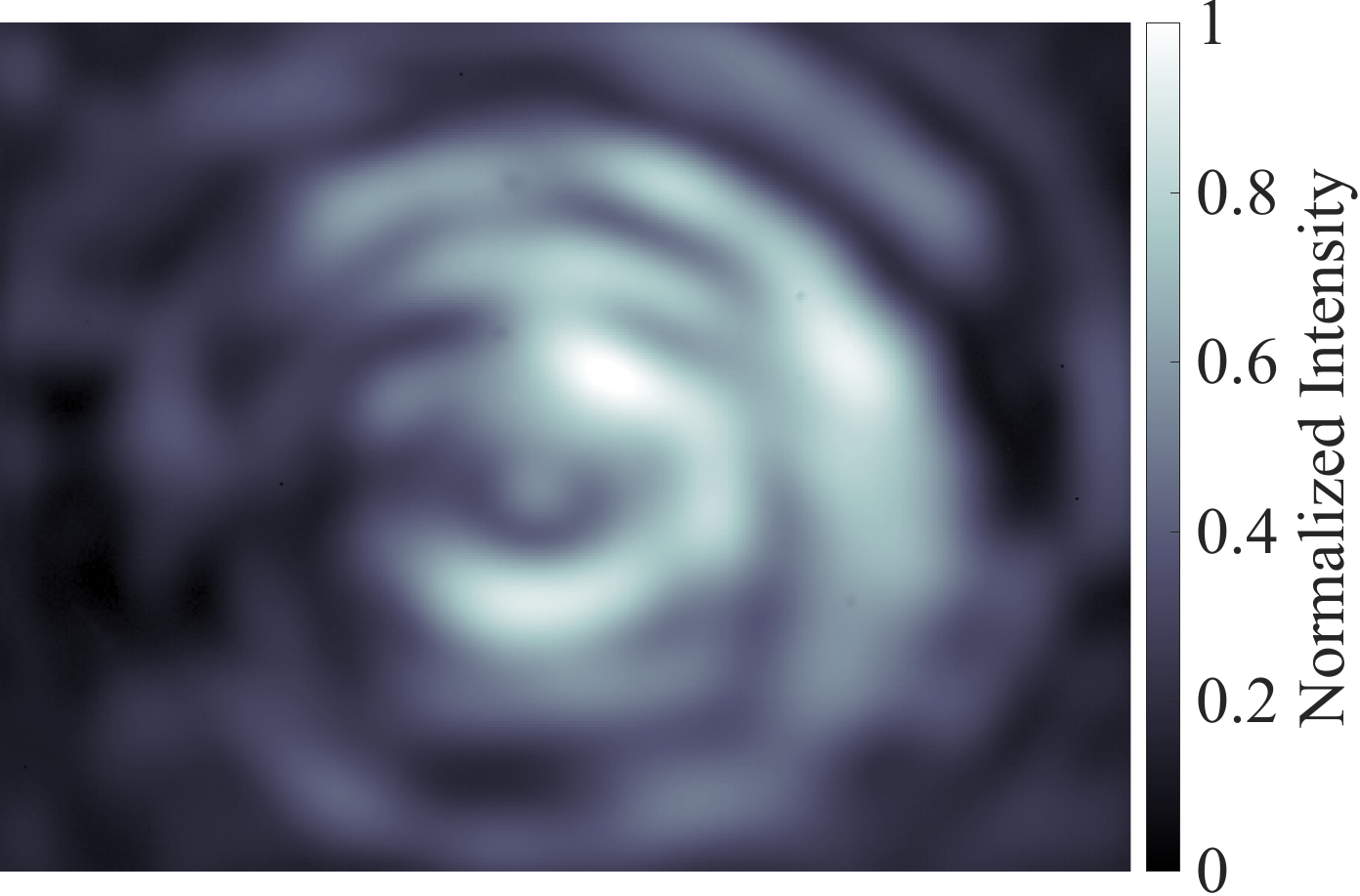}};
    
    \node[white] at (-1.25,-1.3) {{\scriptsize{$C = 0.409$}}};

    \draw[white,ultra thick] 
    (1.4,-1.4) -- (1.84,-1.4);
    \node[white] at (1.6,-1.25) {\bf{\tiny{5 mm}}};
    \end{tikzpicture}
    
    \vspace{-0.3cm}
    \caption{Without Diffuser and Lens}
    \label{p20230823_Diffuser_THzCamera_Beam_1}
    \end{subfigure}
    \begin{subfigure}[b]{0.3\textwidth}
    \begin{tikzpicture}
    \node
    {\adjincludegraphics[width=\textwidth, trim={{0.\width} {0.\height} {.18\width} {0.\height}}, clip]{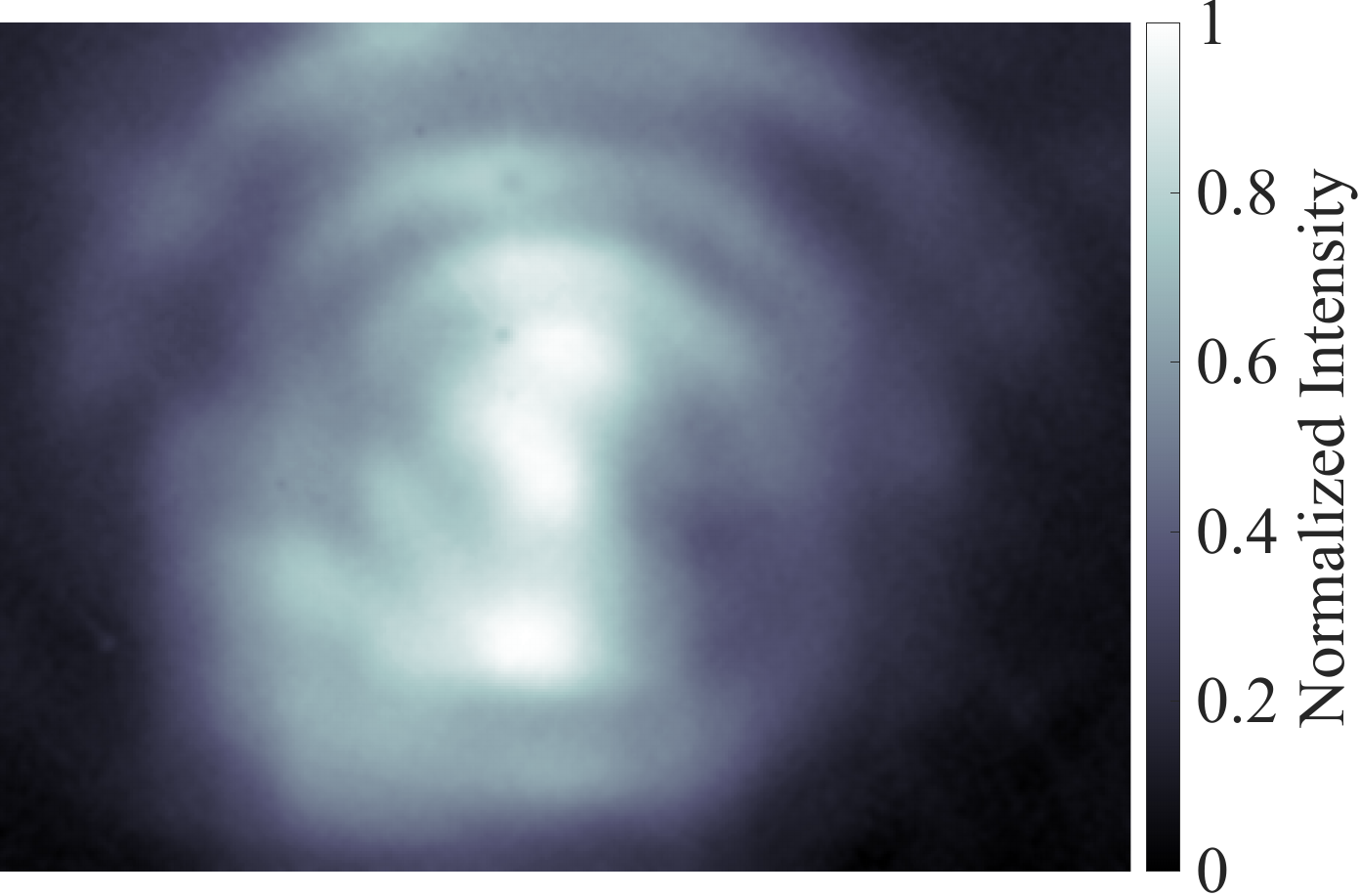}};
    
    \node[white] at (-1.25,-1.3) {{\scriptsize{$C = 0.255$}}};

    \draw[white,ultra thick] 
    (1.4,-1.4) -- (1.84,-1.4);
    \node[white] at (1.6,-1.25) {\bf{\tiny{5 mm}}};
    \end{tikzpicture}
    
    \vspace{-0.3cm}
    \caption{With Diffuser Only}
    \label{p20230823_Diffuser_THzCamera_Beam_2}
    \end{subfigure}
    \begin{subfigure}[b]{0.367\textwidth}
    \begin{tikzpicture}
    \node
    {\includegraphics[width=\textwidth]{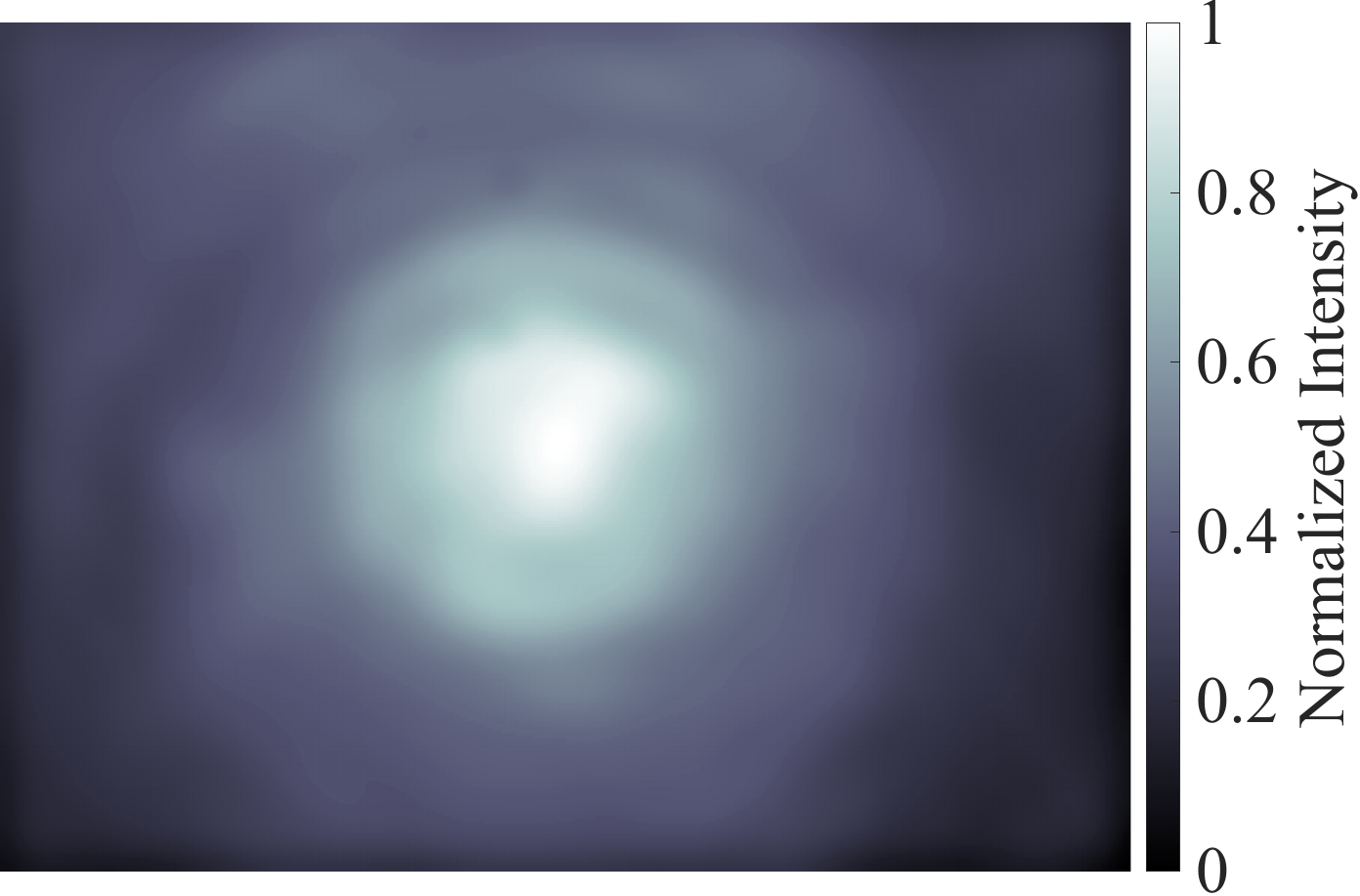}};
    
    \node[white] at (-1.7,-1.3) {{\scriptsize{$C = 0.155$}}};

    \draw[white,ultra thick] 
    (1.05,-1.4) -- (1.49,-1.4);
    \node[white] at (1.25,-1.25) {\bf{\tiny{5 mm}}};
    \end{tikzpicture}
    
    \vspace{-0.3cm}
    \caption{With Diffuser and Lens}
    \label{p20230823_Diffuser_THzCamera_Beam_3}
    \end{subfigure}
    \caption{Illustrations of the illumination pattern and the corresponding speckle contrast value for various experimental configurations:~\subref{p20230823_Diffuser_THzCamera_Beam_1} without diffuser and focusing lens;~\subref{p20230823_Diffuser_THzCamera_Beam_2} with diffuser only (no focusing lens);~\subref{p20230823_Diffuser_THzCamera_Beam_3} with diffuser and focusing lens. }
    \label{p20230823_Diffuser_THzCamera_Beam}
\end{figure}

We now proceed to evaluate the effectiveness of speckle reduction in realistic imaging scenarios using our DaISy on a ``Y''-shaped 3D-printed test object, as depicted in~\cref{p_Img_Y}. 
The strategic selection of a ``Y''-shaped object, with its branching structure with varying angles and contours, makes it ideal for evaluating speckle phenomena across different tilt directions. 
First, the sub-THz imaging result without the diffuser, shown in~\cref{p20230825_Diffuser_THzCamera_Y_1}, reveals that the test object is barely distinguishable, with the image quality significantly impaired by speckle artifacts. 
This highlights the limitations of the system without any speckle mitigation components. 
When only the diffuser is employed, as seen in~\cref{p20230825_Diffuser_THzCamera_Y_2}, the ``Y'' shape of the test object becomes somewhat discernible, yet the image suffers from blurry and lack of clarity. 
Additionally, the edges of the object are distorted, indicating the persistence of diffraction effects. 
The most significant improvement is observed in~\cref{p20230825_Diffuser_THzCamera_Y_3} with the complete configuration (\cref{p_SubTHz_Setup}), incorporating both the diffuser and the focusing lens. 
This configuration markedly improves image contrast, making the edges of the object easily identifiable. 
The clarity achieved in this image starkly contrasts with the previous results, especially compared to the diffuser-only configuration in~\cref{p20230825_Diffuser_THzCamera_Y_2}. 
This comparison underscores the effectiveness of the combined use in achieving speckle-free imaging and delineating finer details. 

\pagebreak

%% Figure 5
\begin{figure}[htbp]
\centering
    \begin{subfigure}[b]{0.22\textwidth}
    \begin{tikzpicture}
    \node
    {\includegraphics[width=\textwidth]{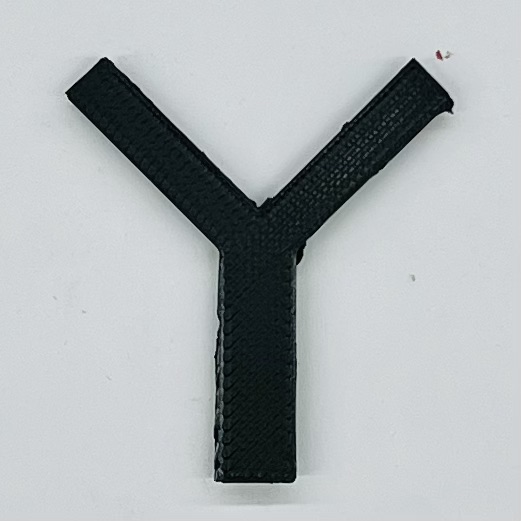}};
    
    \draw[white,ultra thick] 
    (0.9,-1.3) -- (1.34,-1.3);
    \node[white] at (1.1,-1.15) {\bf{\tiny{5 mm}}};
    \end{tikzpicture}
    
    \vspace{-0.25cm}
    \caption{Test Object ``Y''}
    \label{p_Img_Y}
    \end{subfigure}
    \begin{subfigure}[b]{0.222\textwidth}
    \begin{tikzpicture}
    \node
    {\adjincludegraphics[width=\textwidth, trim={{0.\width} {0.\height} {.2\width} {0.\height}}, clip]{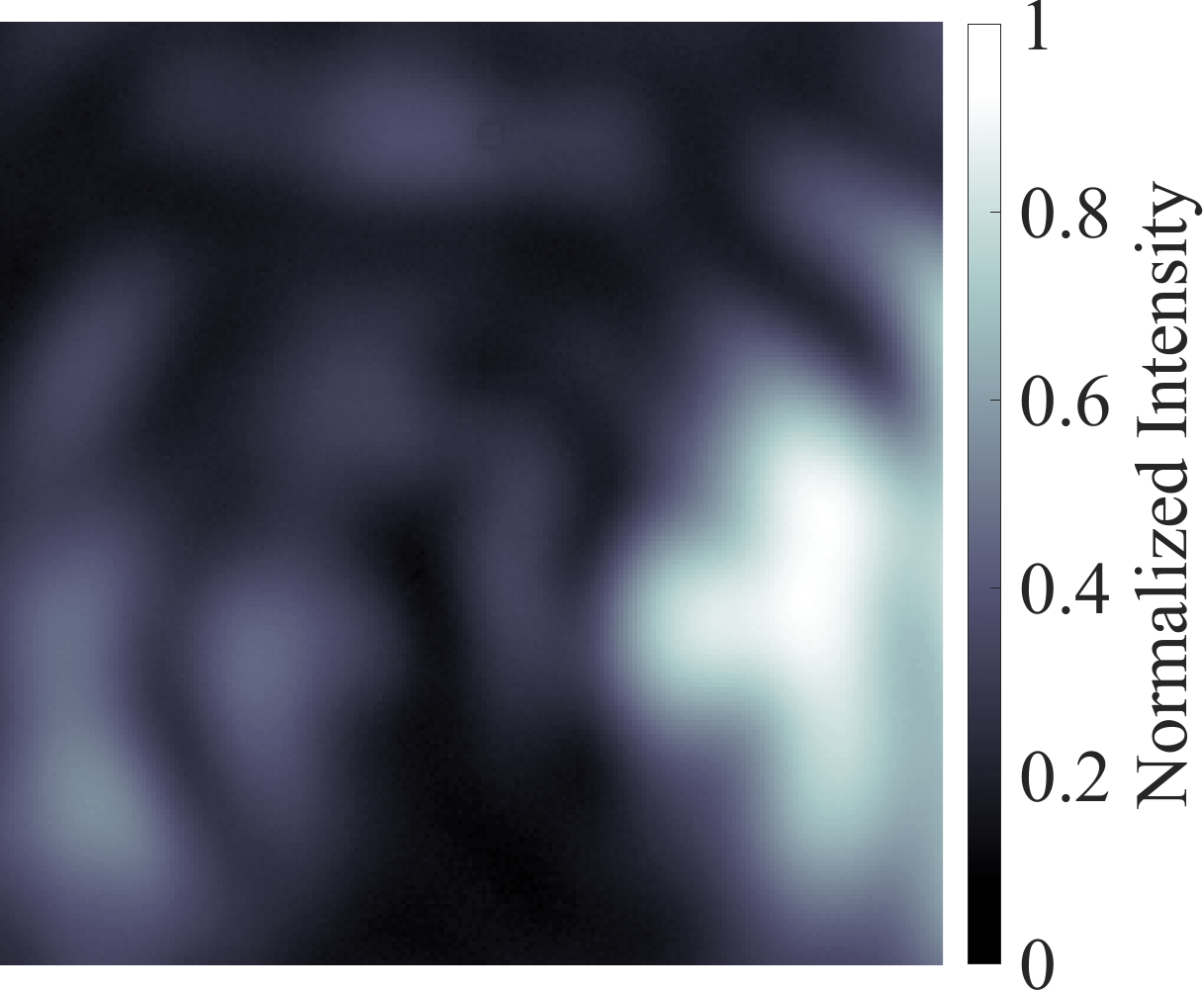}};
    
    \draw[white,ultra thick] 
    (0.9,-1.3) -- (1.34,-1.3);
    \node[white] at (1.1,-1.15) {\bf{\tiny{5 mm}}};
    \end{tikzpicture}
    \vspace{-0.74cm}
    \caption{Without Diffuser}
    \label{p20230825_Diffuser_THzCamera_Y_1}
    \end{subfigure}
    \hspace{-0.15cm}
    \begin{subfigure}[b]{0.222\textwidth}
    \begin{tikzpicture}
    \node
    {\adjincludegraphics[width=\textwidth, trim={{0.\width} {0.\height} {.2\width} {0.\height}}, clip]{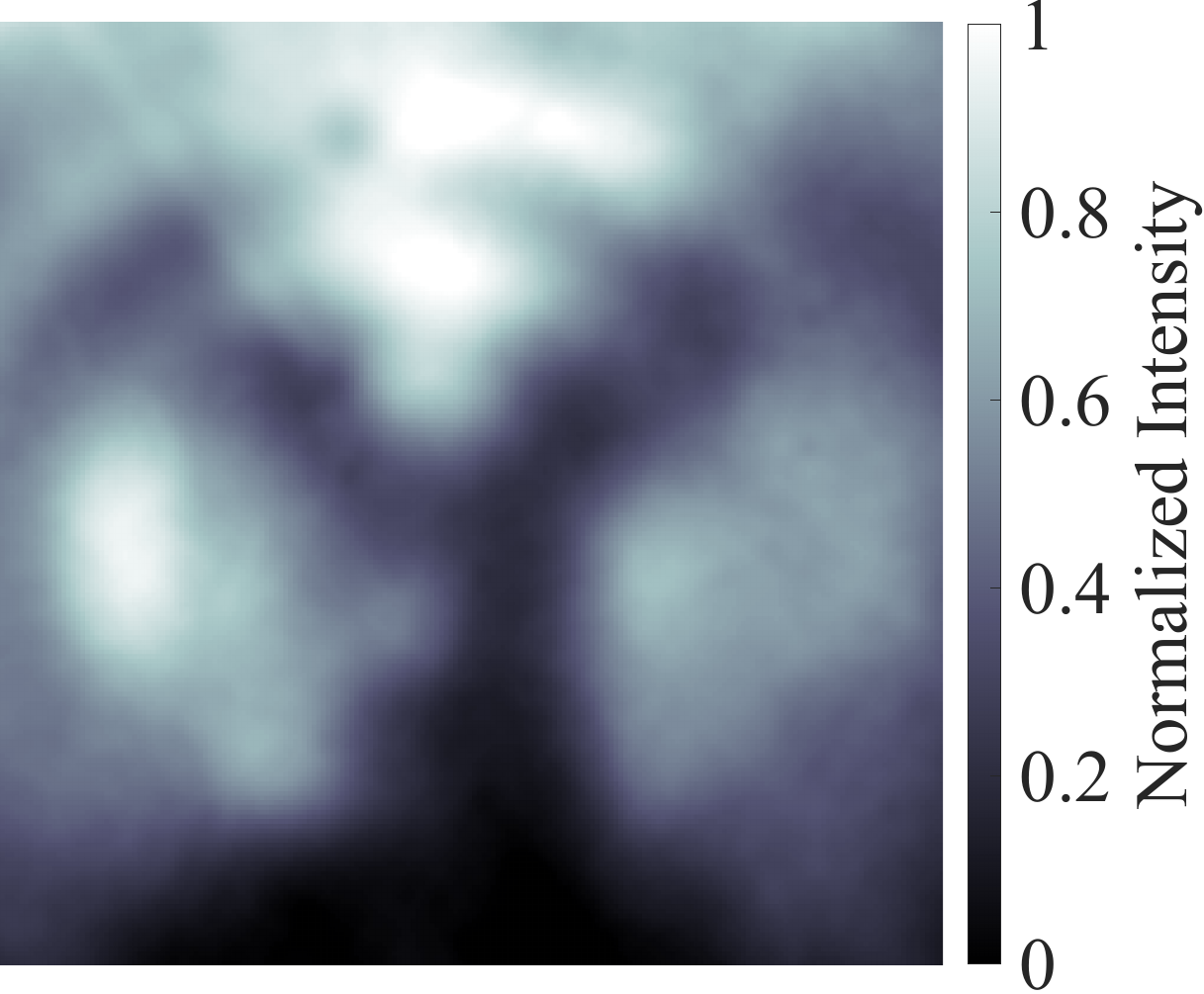}};
    
    \draw[white,ultra thick] 
    (0.9,-1.3) -- (1.34,-1.3);
    \node[white] at (1.1,-1.15) {\bf{\tiny{5 mm}}};
    \end{tikzpicture}
    \vspace{-0.74cm}
    \caption{With Diffuser Only}
    \label{p20230825_Diffuser_THzCamera_Y_2}
    \end{subfigure}
    \hspace{-0.15cm}
    \begin{subfigure}[b]{0.278\textwidth}
    \begin{tikzpicture}
    \node
    {\adjincludegraphics[width=\textwidth, trim={{0.\width} {0.\height} {0.\width} {0.\height}}, clip]{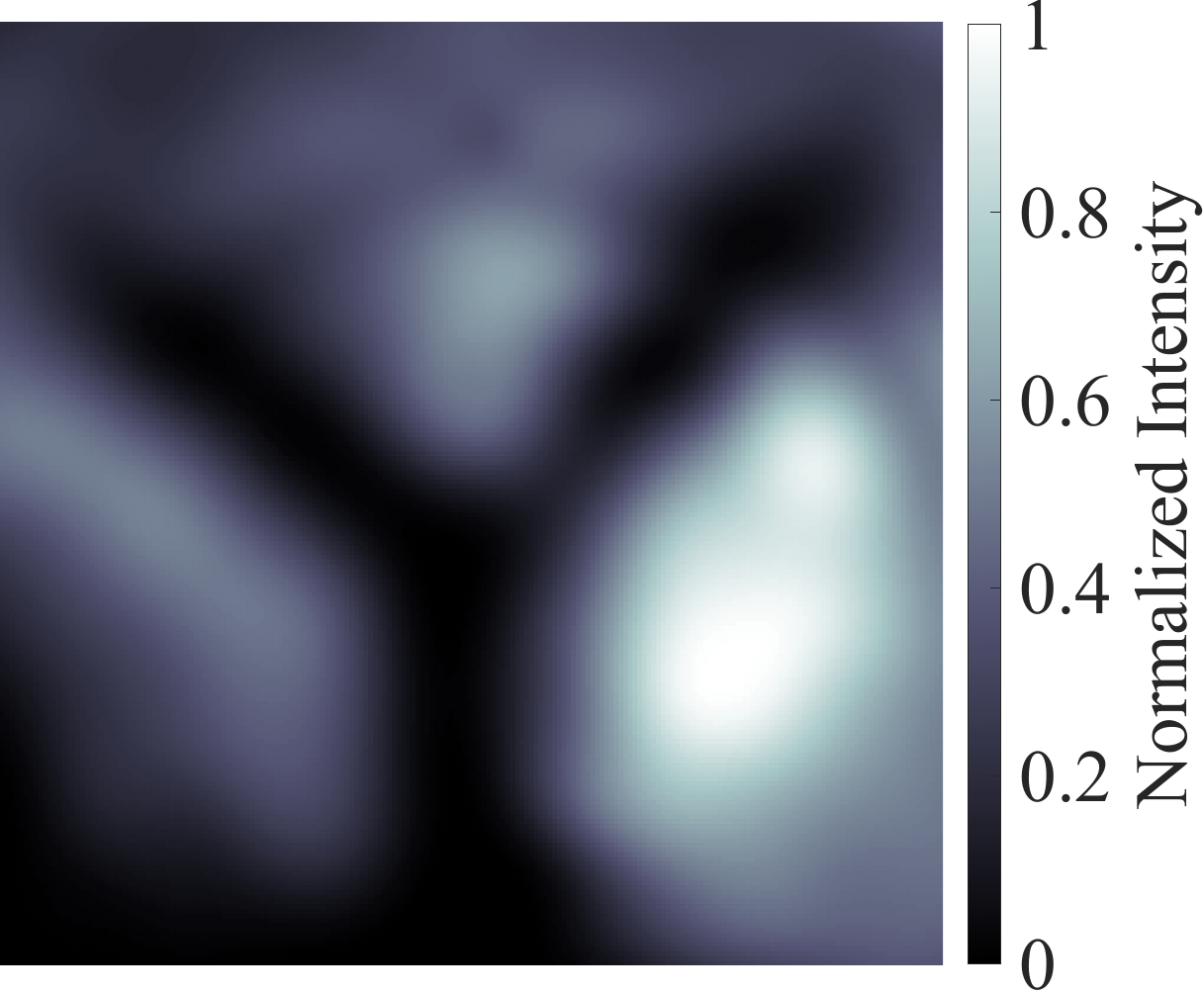}};
    
    \draw[white,ultra thick] 
    (0.55,-1.3) -- (0.99,-1.3);
    \node[white] at (0.75,-1.15) {\bf{\tiny{5 mm}}};
    \end{tikzpicture}
    \vspace{-0.74cm}
    \caption{With Diffuser and Lens}
    \label{p20230825_Diffuser_THzCamera_Y_3}
    \end{subfigure}
    \caption{Illustrations of~\subref{p_Img_Y} the ``Y''-shaped 3D-printed test object for various experimental configurations:~\subref{p20230825_Diffuser_THzCamera_Y_1} without diffuser and focusing lens;~\subref{p20230825_Diffuser_THzCamera_Y_2} with diffuser only (no focusing lens);~\subref{p20230825_Diffuser_THzCamera_Y_3} with diffuser and focusing lens. }
    \label{p20230825_Diffuser_THzCamera_Y}
\end{figure}

\color{reblack}
To quantify the improved image quality from our DaISy system, we employ a comprehensive set of quality metrics outlined in~\cref{tab:metrics}. 
\cref{tab:metrics} showcases the performance of~\cref{p20230825_Diffuser_THzCamera_Y_1}-\subref{p20230825_Diffuser_THzCamera_Y_3} by our DaISy across different configurations: without both a diffuser and a focusing lens, using only a diffuser, and combining both a diffuser and a focusing lens. 
These metrics assess image quality from two perspectives: full-reference quality metrics and no-reference quality metrics. 
For full-reference quality metrics, we compare the acquired images against an ideal reference image, which is derived from the 3D printing file of the ``Y''-shaped object. 
All full reference quality metrics, namely the peak signal-to-noise ratio (PSNR), the structural similarity index measure (SSIM), and the mean-squared error (MSE), show progressive improvements. 
In detail, PSNR jumps significantly from 2.3095 to 7.4017 with the full DaISy setup, indicating a cleaner signal with less noise. 
SSIM increases from 0.2256 to 0.5858, reflecting improved clarity and structural accuracy. 
MSE decreases from 0.5876 to 0.1819, pointing to a lower error between the acquired and reference images. 
Additionally, in a complementary manner, we incorporate no-reference quality metrics, namely the blind/referenceless image spatial quality evaluator (BRISQUE) and the natural image quality evaluator (NIQE), to evaluate the images without a reference image. 
BRISQUE scores decrease from 66.9749 to 63.4722, denoting images that appear more natural and less distorted. 
NIQE scores also improve from 6.5987 to 6.2490, suggesting a quality closer to what the human eye perceives as natural. 
By combining these metrics, we offer a comprehensive evaluation of image quality, highlighting the capability of DaISy to produce detailed and clear images of the ``Y''-shaped specimen without speckle artifacts. 
For a detailed explanation of the quality metrics used in our work, please refer to~\hyperref[supp]{Supplement 1}~\cite{2004_Wang_SSIM_IQM, 2009_Thung_IQM, 2010_Hore_IQM, Mittal_2012_BRISQUE, Mittal_2013_NIQE}. 

%% PSNR, SSIM, MSE, BRISQUE, NIQE
\begin{table}[htbp]
\caption{Comparison metrics for the image quality of~\cref{p20230825_Diffuser_THzCamera_Y_1}-\subref{p20230825_Diffuser_THzCamera_Y_3} of the ``Y''-shaped test object by the full-reference perspective: peak signal-to-noise ratio (PSNR), structural similarity index measure (SSIM), and mean-squared error (MSE), as well as the no-reference perspective: blind/referenceless image spatial quality evaluator (BRISQUE) and naturalness image quality evaluator (NIQE) where $\uparrow$ ($\downarrow$): higher (lower) is better. }
\vspace{-0.5cm}
\label{tab:metrics}
\begin{center}
\begin{adjustbox}{width=.98\textwidth}
    \begin{tabular}{crc|c|c}
    \toprule
    \multicolumn{2}{c}{}                              & Without Diffuser & With Diffuser Only & With Diffuser and Lens \\ \hline
    \multirow{3}{*}{Full-Reference} & PSNR ($\uparrow$) & 2.3095 & 4.0357 & \textbf{7.4017} \\
    & SSIM ($\uparrow$) & 0.2256 & 0.3844 & \textbf{0.5858} \\
    & MSE ($\downarrow$) & 0.5876 & 0.3948 & \textbf{0.1819} \\ \hline
    \multirow{2}{*}{No-Reference} & BRISQUE ($\downarrow$) & 66.9749 & 65.4817 & \textbf{63.4722} \\
    & NIQE ($\downarrow$) & 6.5987 & 6.5692 & \textbf{6.2490} \\
    \bottomrule
    \end{tabular}
    \end{adjustbox}
    \end{center}
    \end{table}

\color{black}
Next, for evaluating imaging resolution with DaISy, we utilize the scaled USAF-1951 resolution test, a standard benchmark in imaging system assessment. 
As illustrated in~\cref{p_Img_USAF}, this test features line pairs varying from 11 mm (top) to 7 mm (bottom), providing a comprehensive test of the resolving capabilities for DaISy. 
For the first configuration, the imaging results without using a diffuser or a focusing lens, shown in~\cref{p20230827_Diffuser_THzCamera_USAF_1}, display significant distortion. 
The larger line pairs, measuring 11 mm, appear blurred, achieving a resolution of only 0.9 line pairs per centimeter (LP/cm). 
A considerable improvement in resolution is observed when employing only the diffuser, as depicted in~\cref{p20230827_Diffuser_THzCamera_USAF_2}. 
In this configuration, the system is able to resolve line pairs down to 9 mm, translating to a resolution of 1.11 LP/cm. 
However, it falls short of distinguishing the finer details of smaller line pairs, indicating the need for further improvement in the resolving ability. 
The most significant advancement is achieved with the integration of both the diffuser and focusing lens, as shown in~\cref{p20230827_Diffuser_THzCamera_USAF_3}. 
This complete configuration enables DaISy to perform at its optimal capacity, successfully resolving line pairs as small as 7 mm. 
This corresponds to an improved resolution of 1.428 LP/cm, clearly demonstrating the improved capabilities in achieving speckle-free imaging with superior resolution. 

%% Figure 6
\begin{figure}[htbp]
\centering
    \begin{subfigure}[b]{0.238\textwidth}
    \hspace{-0.3cm}
    \begin{tikzpicture}
    \node
    {\includegraphics[width=\textwidth]{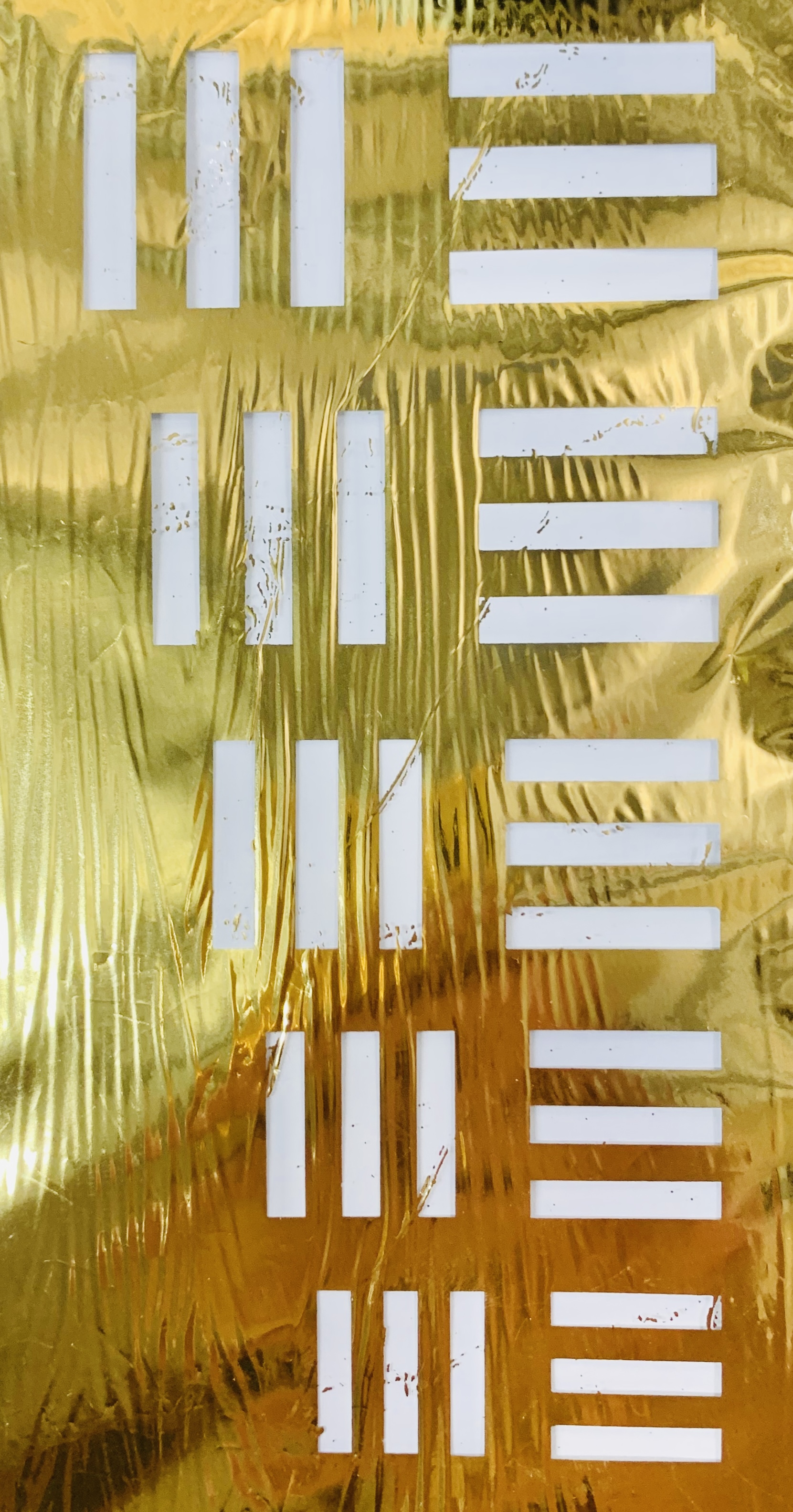}}; % 180 x 345
    
    \draw[white, very thick] 
    (-1.27,1.75) -- (-0.85,1.75);
    \node[white] at (-1.05,1.6) {\bf{\tiny{11 mm}}};

    \draw[white, very thick] 
    (-0.98,0.4) -- (-0.61,0.4);
    \node[white] at (-0.76,0.25) {\bf{\tiny{10 mm}}};

    \draw[white, very thick] 
    (-0.75,-0.83) -- (-0.4,-0.83);
    \node[white] at (-0.53,-0.98) {\bf{\tiny{9 mm}}};

    \draw[white, very thick] 
    (-0.53,-1.92) -- (-0.21,-1.92);
    \node[white] at (-0.3,-2.05) {\bf{\tiny{8 mm}}};

    \draw[white, very thick] 
    (-0.32,-2.84) -- (-0.04,-2.84);
    \node[white] at (-0.1,-2.95) {\bf{\tiny{7 mm}}};
    \end{tikzpicture}
    \vspace{-0.6cm}
    \caption{Resolution Test Object}
    \label{p_Img_USAF}
    \end{subfigure}
    \begin{subfigure}[b]{0.22\textwidth}
    \adjincludegraphics[width=\textwidth, trim={{0.\width} {0.\height} {.2\width} {0.\height}}, clip]{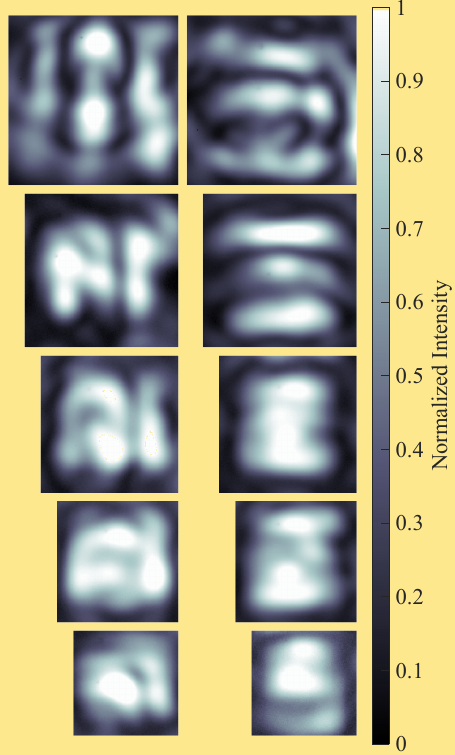}
    \vspace{-0.5cm}
    \caption{Without Diffuser}
    \label{p20230827_Diffuser_THzCamera_USAF_1}
    \end{subfigure}
    \,
    \begin{subfigure}[b]{0.22\textwidth}
    \adjincludegraphics[width=\textwidth, trim={{0.\width} {0.\height} {.2\width} {0.\height}}, clip]{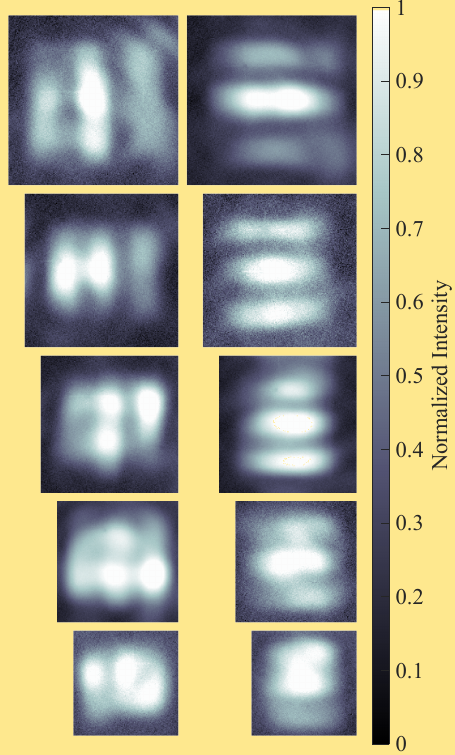}
    \vspace{-0.5cm}
    \caption{With Diffuser Only}
    \label{p20230827_Diffuser_THzCamera_USAF_2}
    \end{subfigure}
    \,
    \begin{subfigure}[b]{0.22\textwidth}
    \adjincludegraphics[width=\textwidth, trim={{0.\width} {0.\height} {.2\width} {0.\height}}, clip]{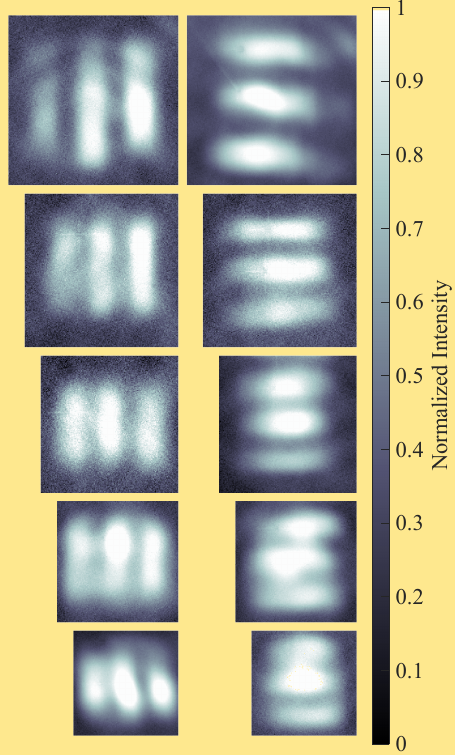}
    \vspace{-0.5cm}
    \caption{With Diffuser and Lens}
    \label{p20230827_Diffuser_THzCamera_USAF_3}
    \end{subfigure}
    \begin{subfigure}[b]{0.055\textwidth}
    \adjincludegraphics[width=\textwidth, trim={{0.8\width} {0.\height} {0.\width} {0.\height}}, clip]{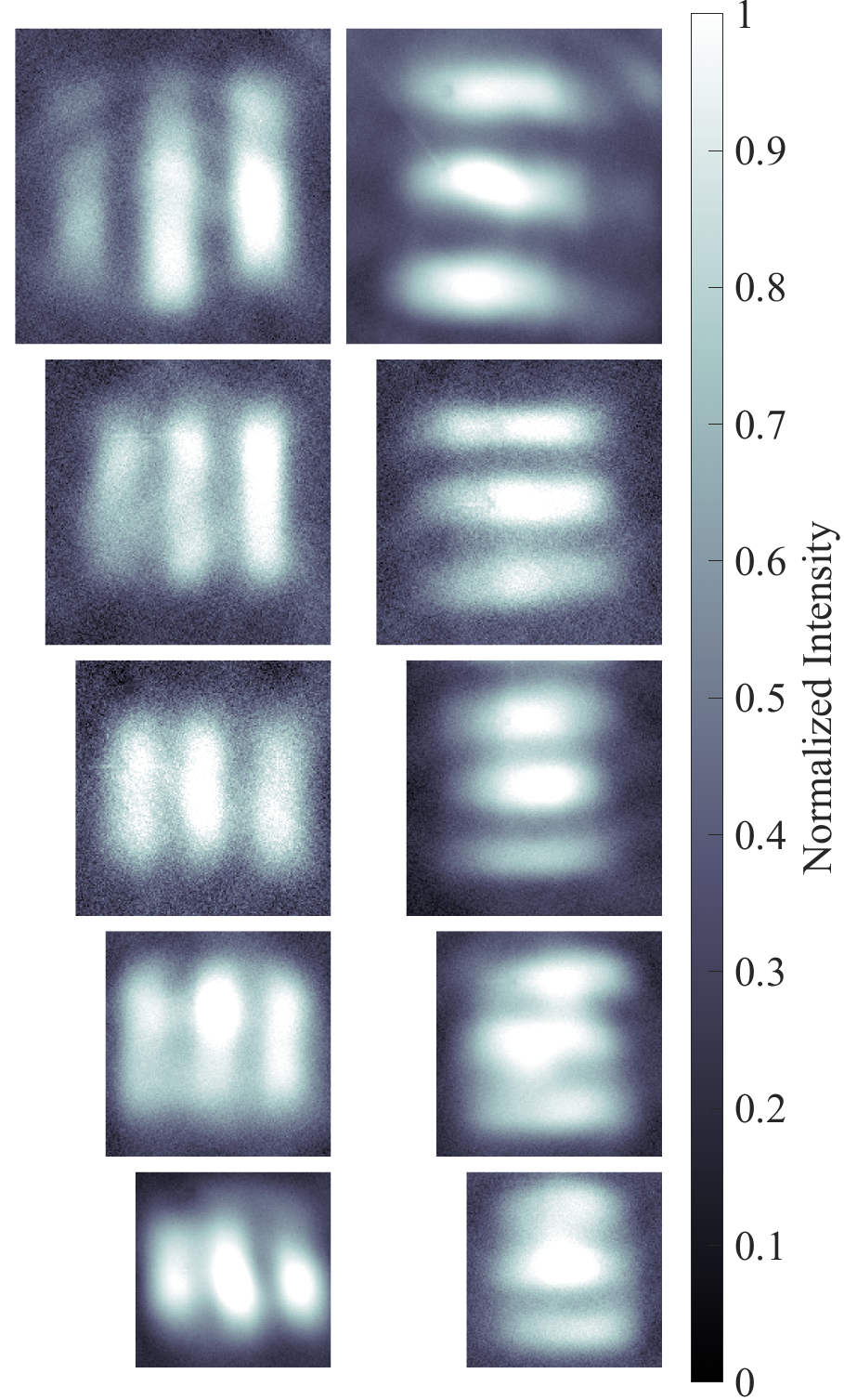}
    \vspace{0.08cm}
    % \caption{With Diffuser and Lens}
    % \label{p20230827_Diffuser_THzCamera_USAF_3}
    \end{subfigure}
    \caption{Illustrations of~\subref{p_Img_USAF} scale USAF-1951 resolution test object, where the scale object displaying resolution ranging from 0.9 LP/cm to 1.428 LP/cm, fabricated by a gold stamping printing process, for various experimental configurations:~\subref{p20230827_Diffuser_THzCamera_USAF_1} without diffuser and focusing lens;~\subref{p20230827_Diffuser_THzCamera_USAF_2} with diffuser only (no focusing lens);~\subref{p20230827_Diffuser_THzCamera_USAF_3} with diffuser and focusing lens. }
    \label{p20230827_Diffuser_THzCamera_USAF}
\end{figure}
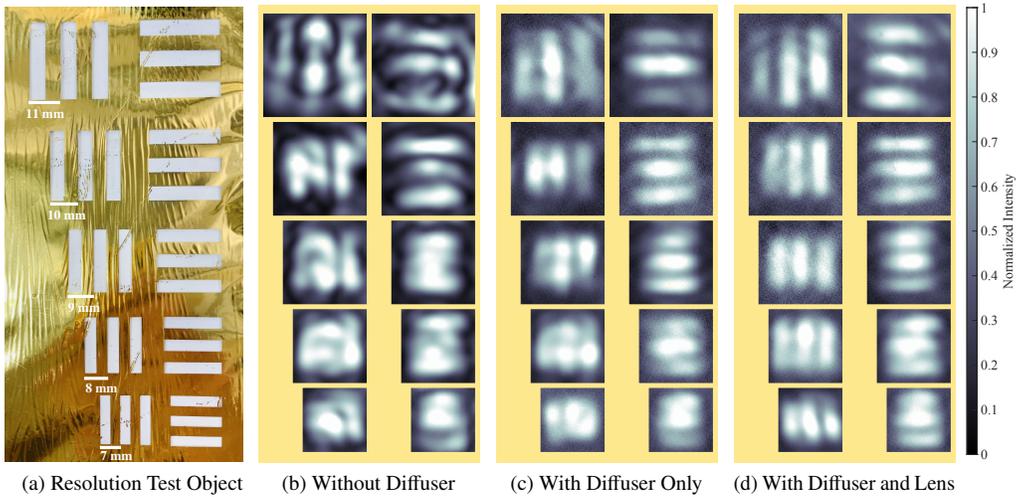

Last, the capabilities of our DaISy extend beyond laboratory testing and find practical applications in real-world scenarios. 
Particularly in the scenario of security scanning, sub-THz imaging stands out due to its unique ability to peer inside concealed objects with exceptional precision. 
To illustrate its practical utility, we conduct a demonstration where a small sculpture knife is concealed in the pocket of a standard nylon jacket, as exemplified in~\cref{p_Img_KnifeInCloth}. 
The optical image of this concealed sculpture knife is shown in~\cref{p_Img_Knife}, akin to that of a pen. 
Leveraging the DaISy, we perform sub-THz imaging of the jacket, and the results, displayed in~\cref{p20230826_Diffuser_THzCamera_Knife_1}, are nothing short of impressive. 
The semblance of the sculpture knife is easy to identify and profile despite its diminutive size of approximately 5 mm. 
Remarkably, even when concealed within the fabric of the jacket, our system offers unparalleled visibility. 
These outcomes underscore the potential of DaISy as a valuable tool in security scanning and inspections, where the precise detection of concealed items is of paramount importance. 

%% Figure 7
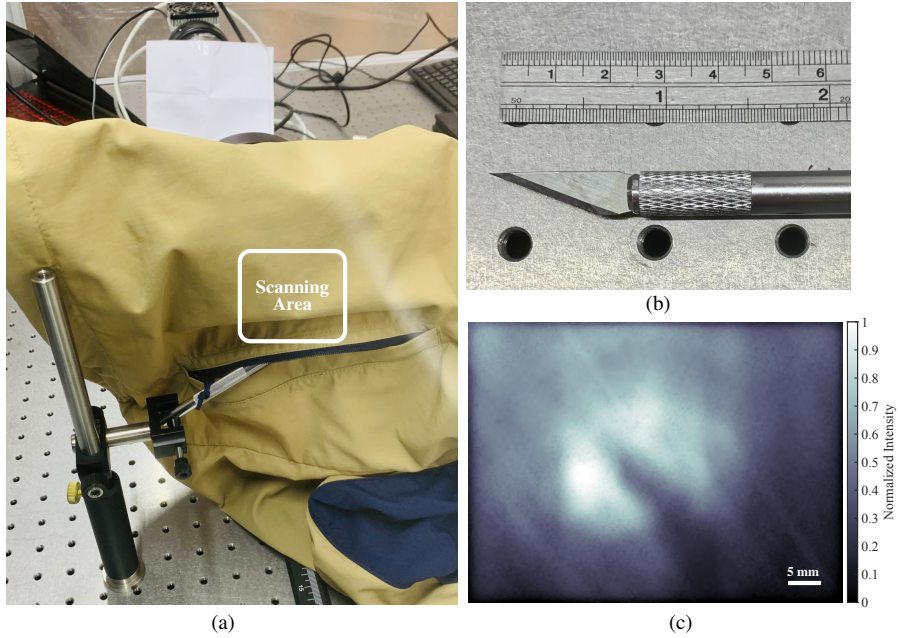
\begin{figure}[htbp]
\centering
    \begin{minipage}[b]{0.45\textwidth}
    \begin{subfigure}{\textwidth}
    \begin{tikzpicture}
    \node
    {\includegraphics[width=\textwidth]{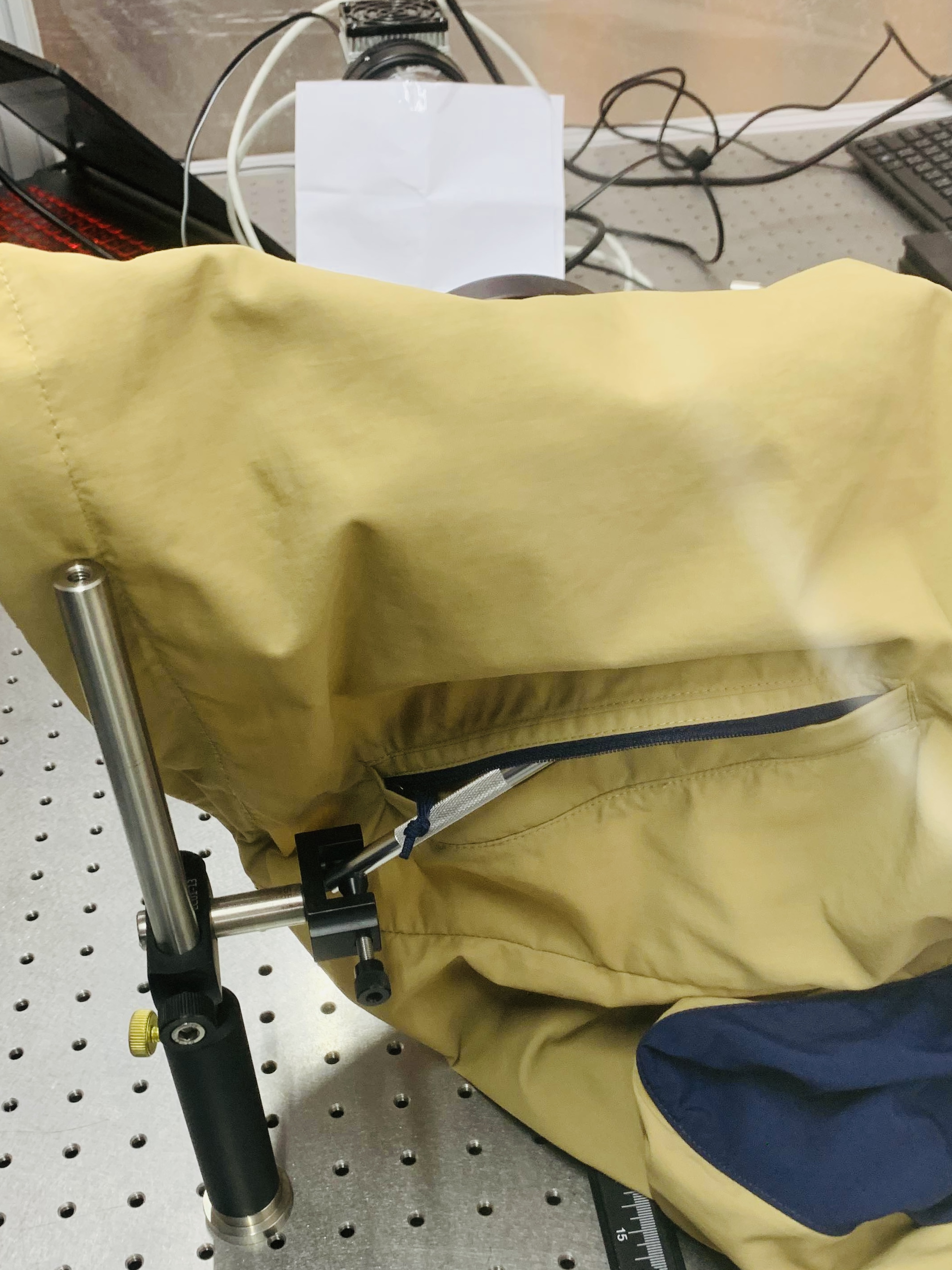}};
    
    \draw[white, ultra thick, rounded corners] 
    (1.5,-0.5) rectangle (0.1,0.7);
    \node[white] at (0.8,0.2) {\bf{\scriptsize{Scanning}}};
    \node[white] at (0.8,-0.0) {\bf{\scriptsize{Area}}};
    \end{tikzpicture}
    \vspace{-0.65cm}
    \caption{}
    \label{p_Img_KnifeInCloth}
    \end{subfigure}
    \end{minipage}
    \,
    \begin{minipage}[b]{0.45\textwidth}
    \begin{subfigure}[b]{.85\textwidth}
    \includegraphics[width=\textwidth]{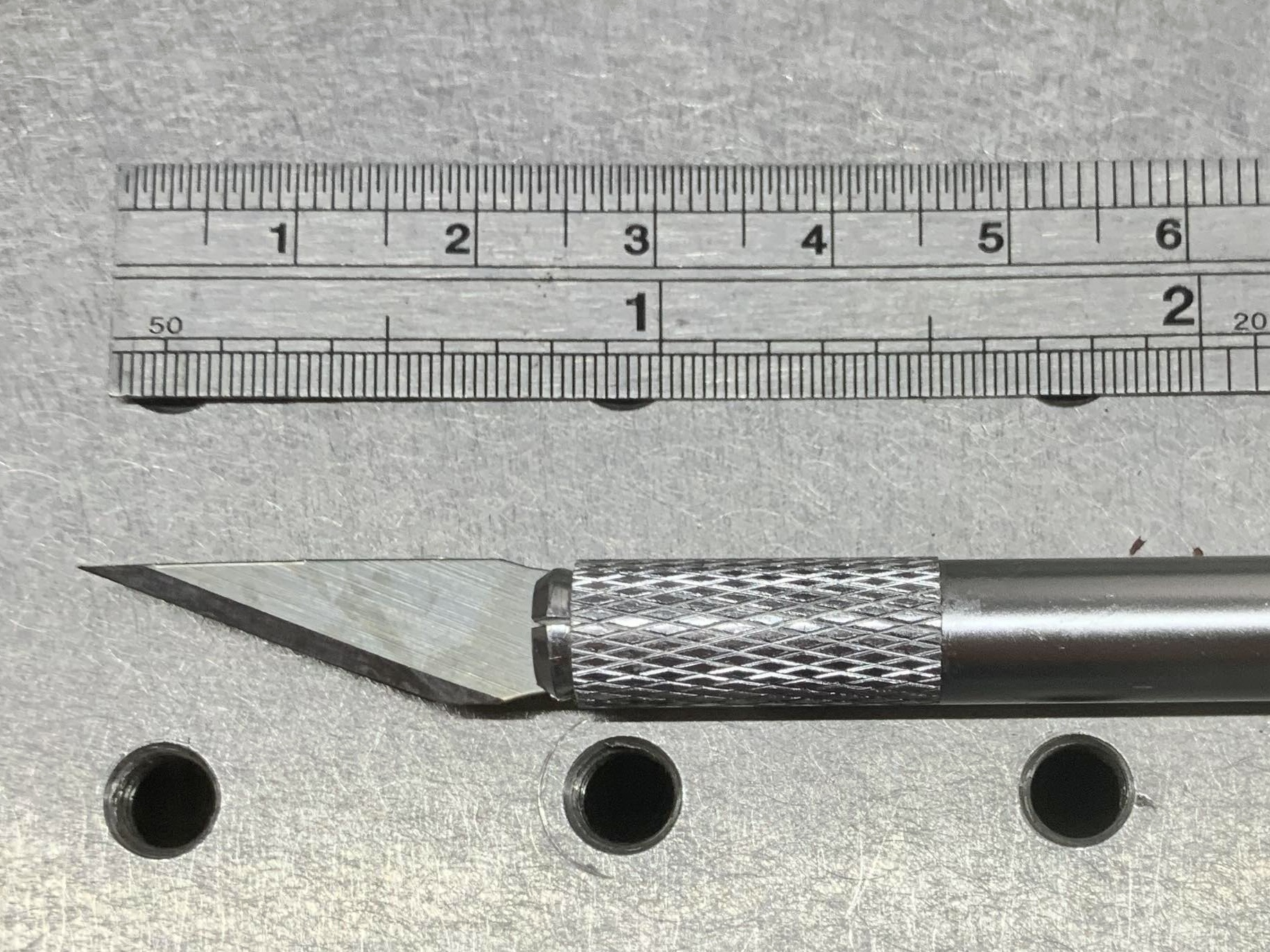}
    \vspace{-0.6cm}
    \caption{}
    \vspace{-0.15cm}
    \label{p_Img_Knife}
    \end{subfigure}
    \\
    \begin{subfigure}[b]{.95\textwidth}
    \hspace{-0.15cm}
    \begin{tikzpicture}
    \node
    {\includegraphics[width=\textwidth]{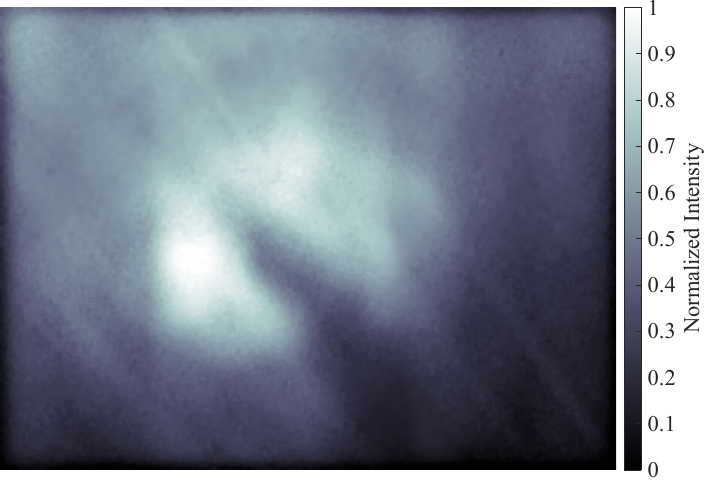}};
    
    \draw[white,ultra thick] 
    (1.4,-1.6) -- (1.84,-1.6);
    \node[white] at (1.60,-1.45) {\bf{\tiny{5 mm}}};
    \end{tikzpicture}
    \vspace{-0.7cm}
    \caption{}
    \label{p20230826_Diffuser_THzCamera_Knife_1}
    \end{subfigure}
    \end{minipage}
    \caption{DaISy demo in security scanning application:~\subref{p_Img_KnifeInCloth} the concealed sculpture knife within a nylon jacket pocket;~\subref{p_Img_Knife} the isolated sculpture knife;~\subref{p20230826_Diffuser_THzCamera_Knife_1} visualization of the sculpture knife using our DaISy with diffuser and focusing lens. }
    \label{p20230826_Diffuser_THzCamera_Knife}
\end{figure}

\section{Conclusion}

The development and implementation of DaISy represent a major advancement in addressing the speckle artifacts in coherent sub-THz imaging. 
At its core, this advancement is driven by integrating a specially designed THz diffuser, made from easily malleable and cost-effective epoxy resin (AB glue), and complemented by a focusing lens. 
This strategic combination effectively breaks the beam coherence, significantly reducing speckle artifacts that arise from the undesired diffraction phenomenon. 
In this work, we establish a theoretical framework in coherence theory to facilitate the design and validation of THz components. 
This framework enables a systematic and quantitative assessment of speckle phenomena associated with different coherence levels. 
Initially, we developed this framework in a configuration employing only a diffuser to analyze speckle contrast. 
Subsequently, we expand this coherence framework to incorporate the wave propagation behavior resulting from the introduction of a diffuser-lens configuration. 
These comprehensive analyses aim to evaluate the efficacy of speckle mitigation, revealing significant enhancements in both image quality and resolution. 
The achievement of a 1.428 LP/cm in resolution is a testament to the efficacy of this approach, pushing the boundaries of a few-cycle resolution in sub-THz imaging. 
Moreover, the practical application of DaISy, particularly demonstrated in security scanning, highlights the versatile potential of sub-THz waves in various fields. 
The ability to discern concealed objects with improved clarity and precision underlines the value of DaISy in real-world scenarios, opening doors to its use in diverse areas ranging from non-invasive sensing and imaging to industrial inspection to bioinformatics. 

However, one of the limitations of the current demonstrated system is it can only capture the spatial information of objects, lacking the capacity to collect multi-dimensional information for further analysis. 
Recent advancements in machine learning techniques have shown great promise in extracting phase information through amplitude information~\cite{Hao_2021_Auto-focusing, Yichen_2018_Extended, Yuan_2019_FourierImaging, Xiang_2022_amplitudephase}. 
By incorporating data-driven methodologies within the DaISy architecture, it is expected to have real-time, highly precise, cost-effective, and low-complexity THz imaging systems in the near future. 
This advancement holds great promise for expanding the application scopes of THz imaging, akin to its neighboring electromagnetic wave bands. 

% \pagebreak

\appendix
\section{Derivation of~\cref{eqn:avg_MGF}}
\label{appendix:derivation}

\begin{theorem}
\label{thm_avg_cf}
    If $\phi$ is a random variable of a Gaussian (normal) distribution $\phi \sim \mathcal{N}(\mu_{\phi} = 0, \sigma_{\phi}^{2})$ with zero mean $\mu_{\phi}$ and standard deviation $\sigma_{\phi}$ (variance $\sigma_{\phi}^{2}$), there exists
    \[
    \mathbb{E}\{e^{j \Delta \phi}\}
    \, \textcolor{reblack}{=} \, 
    \langle
    e^{j \Delta \phi}
    \rangle
    \, \textcolor{reblack}{=} \, 
    \overline{e^{j \Delta \phi}}
    =
    e^{- \sigma^{2}_{\phi} [1 - \mu_{\phi}(\Delta \alpha, \Delta \beta)]}. 
    \]
\end{theorem}

\begin{proof}
    
    Since $\phi$ is a Gaussian (normal) random variable $\phi \sim \mathcal{N}(\mu_{\phi} = 0, \sigma_{\phi}^{2})$, its Gaussian probability density function (PDF) can be formulated as
    \begin{eqnarray*}
        p(\phi)
        \, \textcolor{reblack}{=} \, 
        \frac{1}{\sqrt{2 \pi \sigma_{\phi}^{2}}} 
        e^{-\frac{(\phi - \mu_{\phi})^{2}}{2 \sigma_{\phi}^{2}}}
        \xlongequal{\mu_{\phi} = 0}
        \frac{1}{\sqrt{2 \pi \sigma_{\phi}^{2}}} 
        e^{-\frac{\phi^{2}}{2 \sigma_{\phi}^{2}}}. 
    \end{eqnarray*}
    Note that $\Delta \phi$ is also a Gaussian random variable. 

    The characteristic function of above PDF $p(\phi)$ is the complex conjugate of the continuous Fourier transform of $p(\phi)$
    \begin{eqnarray*}
    \begin{aligned}
        \varphi_{\phi} 
        \, \textcolor{reblack}{=} \, 
        \mathbb{E}\{e^{j \phi}\}
        \, \textcolor{reblack}{=} \, 
        \langle
        e^{j \phi}
        \rangle
        \, \textcolor{reblack}{=} \, 
        \overline{e^{j \phi}}
        & =
        \frac{1}{\sqrt{2 \pi \sigma_{\phi}^{2}}} 
        \int_{-\infty}^{\infty}
        \left[
        e^{-\frac{(\phi - \mu_{\phi})^{2}}{2 \sigma_{\phi}^{2}}}
        e^{j \phi}
        \right]
        \, \mathrm{d}\phi
        \\
        & %\,\, 
        = 
        \frac{1}{\sqrt{2 \pi \sigma_{\phi}^{2}}} 
        \int_{-\infty}^{\infty}
        \left\{
        e^{
        -\frac{1}{2\sigma_{\phi}^{2}}
        \left[
        \phi - (\mu_{\phi} + j\sigma_{\phi}^{2})
        \right]^{2}
        -
        \frac{1}{2} \sigma_{\phi}^{2}
        +
        j \mu_{\phi}
        }
        \right\}
        \, \mathrm{d}\phi
        \\
        & %\,\, 
        =
        e^{-\frac{1}{2}\sigma_{\phi}^{2} + j\mu_{\phi}}
        \frac{1}{\sqrt{2 \pi \sigma_{\phi}^{2}}} 
        \int_{-\infty}^{\infty}
        \left\{
        e^{
        -\frac{1}{2\sigma_{\phi}^{2}}
        \left[
        \phi - (\mu_{\phi} + j\sigma_{\phi}^{2})
        \right]^{2}
        }
        \right\}
        \, \mathrm{d}\phi
        \\
        & %\,\, 
        \xlongequal{\mathrm{by \, Gaussian \, integral}}
        e^{-\frac{1}{2}\sigma_{\phi}^{2} + j\mu_{\phi}}
        \xlongequal{\mu_{\phi} = 0}
        e^{-\frac{1}{2}\sigma_{\phi}^{2}}. 
    \end{aligned}
    \end{eqnarray*}
    Then, there exists
    \begin{eqnarray*}
        \varphi_{\Delta \phi} 
        \, \textcolor{reblack}{=} \, 
        \mathbb{E}\{e^{j\Delta \phi}\}
        =
        e^{-\frac{1}{2}\sigma_{\Delta \phi}^{2}}
        =
        e^{- \sigma^{2}_{\phi} [1 - \mu_{\phi}(\Delta \alpha, \Delta \beta)]} 
    \end{eqnarray*}
    where
    \begin{eqnarray*}
        \sigma_{\Delta \phi}^{2}
        =
        \mathbb{E}\{(\Delta \phi)^{2}\}
        =
        \mathbb{E}\{\phi_{1}^{2}\}
        +
        \mathbb{E}\{\phi_{2}^{2}\}
        -
        2 \mathbb{E}\{\phi_{1}\phi_{2}\}
        =
        \sigma_{\phi}^{2} 
        +
        \sigma_{\phi}^{2} 
        -
        2 \sigma_{\phi}^{2} \mu_{\phi}(\Delta \alpha, \Delta \beta)
    \end{eqnarray*}
    with
    \begin{eqnarray*}
    \begin{aligned}
        \mathbb{E}\{\phi_{1} \phi_{2}\}
        & \, \textcolor{reblack}{=} \, 
        \mathbb{E}\{\phi_{1}(\alpha, \beta) \phi_{2}(\alpha + \Delta \alpha, \beta + \Delta \beta)\} \\
        & \,\, =
        \mathbb{E}\{\phi_{1}(\alpha, \beta) \phi_{2}(\alpha, \beta)\}
        \mathbb{E}\{\phi_{2}(\Delta \alpha, \Delta \beta)\} \\
        & \,\, =
        \sigma_{\phi}^{2} \mu_{\phi}(\Delta \alpha, \Delta \beta), 
    \end{aligned}
    \end{eqnarray*}
    which is the autocorrelation function
    \begin{eqnarray*}
        \mu_{\phi}(\Delta \alpha, \Delta \beta)
        =
        \frac
        {\mathbb{E}\{\phi_{1}(\alpha, \beta) \phi_{2}(\alpha + \Delta \alpha, \beta + \Delta \beta)\}}
        {\sigma_{\phi}^{2}}. 
    \end{eqnarray*}
    
\end{proof}

% \pagebreak

\begin{backmatter}
\bmsection{Funding}National Science and Technology Council, Taiwan (NSTC 112-2221-E-007-089-MY3). 

\bmsection{Acknowledgments}Y. Zhang is grateful for partial support from the UoL-NTHU Dual PhD Programme. 
The authors thank Mr Pierre Talbot (INO) and their NTHU colleagues, including Professors Yi-Chun Liu, Kai-Ming Feng and Chia-Wen Lin, along with fellows Ms Yun-Ting Tseng, Mr Tsung-Han Wu, Mr Yi-Chun Hung, Mr Seyed Mostafa Latifi, Ms Vinsa Kharisma, Ms Shaghayegh Afshar, Mr Pouya Torkaman, and Mr Chia-Ming Mai for their valuable support and insightful suggestions. 

\bmsection{Disclosures}The authors declare no conflicts of interest. 

\bmsection{Data Availability}Data underlying the results presented in this paper are not publicly available at this time but may be obtained from the authors upon reasonable request. 

\bmsection{Supplemental document}\phantomsection\label{supp}See \href{https://doi.org/10.6084/m9.figshare.25328746}{Supplement 1} for supporting content. 

\end{backmatter}

%%%%%%%%%% If using BibTeX:
\bibliography{2_refs.bib}

\begin{thebibliography}{10}
\newcommand{\enquote}[1]{``#1''}

\bibitem{Koenig_2013_1_Wireless}
S.~Koenig, D.~Lopez-Diaz, J.~Antes, \emph{et~al.},
  \enquote{\href{https://doi.org/10.1038/nphoton.2013.275}{Wireless sub-THz
  communication system with high data rate},} {\protect\JournalTitle{Nature
  photonics}} \textbf{7}, 977--981 (2013).

\bibitem{Kiessling_2013_2_HighPower}
J.~Kie{\ss}ling, I.~Breunig, P.~Schunemann, \emph{et~al.},
  \enquote{\href{https://doi.org/10.1088/1367-2630/15/10/105014}{High power and
  spectral purity continuous-wave photonic THz source tunable from 1 to 4.5 THz
  for nonlinear molecular spectroscopy},} {\protect\JournalTitle{New Journal of
  Physics}} \textbf{15}, 105014 (2013).

\bibitem{Liu_2010_3_Broadband}
J.~Liu, J.~Dai, S.~L. Chin, and X.-C. Zhang,
  \enquote{\href{https://doi.org/10.1038/nphoton.2010.165}{Broadband terahertz
  wave remote sensing using coherent manipulation of fluorescence from
  asymmetrically ionized gases},} {\protect\JournalTitle{Nature Photonics}}
  \textbf{4}, 627--631 (2010).

\bibitem{Shakeel_2021_4_Creating}
P.~M. Shakeel, S.~Baskar, H.~Fouad, \emph{et~al.},
  \enquote{\href{https://doi.org/10.1007/s11036-020-01670-9}{Creating
  collision-free communication in IoT with 6G using multiple machine access
  learning collision avoidance protocol},} {\protect\JournalTitle{Mobile
  Networks and Applications}} \textbf{26}, 969--980 (2021).

\bibitem{Kawase_2003_5_Non}
K.~Kawase, Y.~Ogawa, Y.~Watanabe, and H.~Inoue,
  \enquote{\href{https://doi.org/10.1364/OE.11.002549}{Non-destructive
  terahertz imaging of illicit drugs using spectral fingerprints},}
  {\protect\JournalTitle{Optics express}} \textbf{11}, 2549--2554 (2003).

\bibitem{Zhang_2023_CLEO_CT_VM_EE}
Y.~Zhang, K.~Chen, and S.-H. Yang,
  \enquote{\href{https://opg.optica.org/abstract.cfm?URI=CLEO_SI-2023-JTh2A.104}{Euler-Elastica
  Variational Model for Pulsed Terahertz 3D Imaging},} in \emph{CLEO 2023,}
  (Optica Publishing Group, 2023), p. JTh2A.104.

\bibitem{NAFTALY_2005_Terahertz_transmission}
M.~Naftaly, A.~Foulds, R.~Miles, and A.~Davies,
  \enquote{\href{https://doi.org/10.1007/s10762-004-2033-6}{Terahertz
  transmission spectroscopy of nonpolar materials and relationship with
  composition and properties},} {\protect\JournalTitle{International Journal of
  Infrared and Millimeter Waves}} \textbf{26}, 55--64 (2005).

\bibitem{Walker_2002_10_safety_guidelines}
G.~C. Walker, E.~Berry, N.~N. Zinov'ev, \emph{et~al.},
  \enquote{\href{https://doi.org/10.1117/12.465614}{Terahertz imaging and
  international safety guidelines},} {\protect\JournalTitle{Medical Imaging
  2002: Physics of Medical Imaging}} \textbf{4682}, 683--690 (2002).

\bibitem{Berry_2003_6_safety_issues}
E.~Berry, \enquote{\href{https://doi.org/10.1023/A:1024461313486}{Risk
  perception and safety issues},} {\protect\JournalTitle{Journal of Biological
  Physics}} \textbf{29}, 263--267 (2003).

\bibitem{Clothier_2003_6_Effects}
R.~H. Clothier and N.~Bourne,
  \enquote{\href{https://doi.org/10.1023/A:1024492725782}{Effects of THz
  exposure on human primary keratinocyte differentiation and viability},}
  {\protect\JournalTitle{Journal of Biological Physics}} \textbf{29}, 179--185
  (2003).

\bibitem{Federici_2005_6_THz}
J.~F. Federici, B.~Schulkin, F.~Huang, \emph{et~al.},
  \enquote{\href{https://doi.org/10.1088/0268-1242/20/7/018}{THz imaging and
  sensing for security applications—explosives, weapons and drugs},}
  {\protect\JournalTitle{Semiconductor science and technology}} \textbf{20},
  S266 (2005).

\bibitem{Park_2015_7_IC}
S.-H. Park, J.-W. Jang, and H.-S. Kim,
  \enquote{\href{https://doi.org/10.1088/0960-1317/25/9/095007}{Non-destructive
  evaluation of the hidden voids in integrated circuit packages using terahertz
  time-domain spectroscopy},} {\protect\JournalTitle{Journal of Micromechanics
  and Microengineering}} \textbf{25}, 095007 (2015).

\bibitem{SUN_2011_A_promising}
Y.~Sun, M.~Y. Sy, Y.-X.~J. Wang, \emph{et~al.},
  \enquote{\href{https://doi.org/10.4329/wjr.v3.i3.55}{A promising diagnostic
  method: Terahertz pulsed imaging and spectroscopy},}
  {\protect\JournalTitle{World journal of radiology}} \textbf{3}, 55--65
  (2011).

\bibitem{JUNG_2012_Quantitative}
E.~Jung, H.~J. Choi, M.~Lim, \emph{et~al.},
  \enquote{\href{https://doi.org/10.1364/BOE.3.001110}{Quantitative analysis of
  water distribution in human articular cartilage using terahertz time-domain
  spectroscopy},} {\protect\JournalTitle{Biomedical Optics Express}}
  \textbf{3}, 1110--1115 (2012).

\bibitem{TSENG_2015_Terahertz_Near-field}
T.-F. Tseng, S.-C. Yang, Y.-T. Shih, \emph{et~al.},
  \enquote{\href{https://doi.org/10.1364/OE.23.025058}{Near-field sub-THz
  transmission-type image system for vessel imaging in-vivo},}
  {\protect\JournalTitle{Optics Express}} \textbf{23}, 25058--25071 (2015).

\bibitem{GENTE_2015_Monitoring}
R.~Gente and M.~Koch,
  \enquote{\href{https://doi.org/10.1186/s13007-015-0057-7}{Monitoring leaf
  water content with THz and sub-THz waves},} {\protect\JournalTitle{Plant
  methods}} \textbf{11} (2015).

\bibitem{OYAMA_2009_Sub-terahertz_imaging_of_defects_in_building_blocks}
Y.~Oyama, L.~Zhen, T.~Tanabe, and M.~Kagaya,
  \enquote{\href{https://doi.org/10.1016/j.ndteint.2008.08.002}{Sub-terahertz
  imaging of defects in building blocks},} {\protect\JournalTitle{Ndt \& E
  International}} \textbf{42}, 28--33 (2009).

\bibitem{TZYDYNZHAPOV_2020_security}
G.~Tzydynzhapov, P.~Gusikhin, V.~Muravev, \emph{et~al.},
  \enquote{\href{https://doi.org/10.1007/s10762-020-00683-5}{New real-time
  sub-terahertz security body scanner},} {\protect\JournalTitle{Journal of
  Infrared, Millimeter, and Terahertz Waves}} \textbf{41}, 631--641 (2020).

\bibitem{Yi_2021_THzImagingSystem}
L.~Yi, Y.~Nishida, T.~Sagisaka, \emph{et~al.},
  \enquote{\href{https://doi.org/10.1109/JLT.2021.3092779}{Towards Practical
  Terahertz Imaging System With Compact Continuous Wave Transceiver},}
  {\protect\JournalTitle{Journal of Lightwave Technology}} \textbf{39},
  7850--7861 (2021).

\bibitem{Wang_2022_HighSpeed_THzImaging}
Y.~Wang, L.~Yi, M.~Tonouchi, and T.~Nagatsuma,
  \enquote{\href{https://doi.org/10.3390/photonics9120913}{High-Speed 600
  GHz-Band Terahertz Imaging Scanner System with Enhanced Focal Depth},}
  {\protect\JournalTitle{Photonics}} \textbf{9} (2022).

\bibitem{Al_2012_FET_Image}
R.~Al~Hadi, H.~Sherry, J.~Grzyb, \emph{et~al.},
  \enquote{\href{https://doi.org/10.1109/JSSC.2012.2217851}{A 1 k-pixel video
  camera for 0.7--1.1 terahertz imaging applications in 65-nm CMOS},}
  {\protect\JournalTitle{IEEE Journal of Solid-State Circuits}} \textbf{47},
  2999--3012 (2012).

\bibitem{Dyakonov_1996_FET_Principle}
M.~Dyakonov and M.~Shur,
  \enquote{\href{https://doi.org/10.1109/16.485650}{Detection, mixing, and
  frequency multiplication of terahertz radiation by two-dimensional electronic
  fluid},} {\protect\JournalTitle{IEEE transactions on electron devices}}
  \textbf{43}, 380--387 (1996).

\bibitem{Knap_2002_FET_Experiment_IIIV_Resonant}
W.~Knap, Y.~Deng, S.~Rumyantsev, and M.~Shur,
  \enquote{\href{https://doi.org/10.1063/1.1525851}{Resonant detection of
  subterahertz and terahertz radiation by plasma waves in submicron
  field-effect transistors},} {\protect\JournalTitle{Applied physics letters}}
  \textbf{81}, 4637--4639 (2002).

\bibitem{Knap_2002_FET_Experiment_IIIV_Nonresonant}
W.~Knap, V.~Kachorovskii, Y.~Deng, \emph{et~al.},
  \enquote{\href{https://doi.org/10.1063/1.1468257}{Nonresonant detection of
  terahertz radiation in field effect transistors},}
  {\protect\JournalTitle{Journal of Applied Physics}} \textbf{91}, 9346--9353
  (2002).

\bibitem{Knap_2004_FET_Experiment_Si}
W.~Knap, F.~Teppe, Y.~Meziani, \emph{et~al.},
  \enquote{\href{https://doi.org/10.1063/1.1775034}{Plasma wave detection of
  sub-terahertz and terahertz radiation by silicon field-effect transistors},}
  {\protect\JournalTitle{Applied Physics Letters}} \textbf{85}, 675--677
  (2004).

\bibitem{Li_2023_FPA_Review}
X.~Li, J.~Li, Y.~Li, \emph{et~al.},
  \enquote{\href{https://doi.org/10.1038/s41377-023-01278-0}{High-throughput
  terahertz imaging: progress and challenges},} {\protect\JournalTitle{Light:
  Science \& Applications}} \textbf{12}, 233 (2023).

\bibitem{Ori_2014_Non-invasive}
O.~Katz, P.~Heidmann, M.~Fink, and S.~Gigan,
  \enquote{\href{https://doi.org/10.1038/nphoton.2014.189}{Non-invasive
  single-shot imaging through scattering layers and around corners via speckle
  correlations},} {\protect\JournalTitle{Nature photonics}} \textbf{8},
  784--790 (2014).

\bibitem{Ma_2018_ma2018multimode}
R.~Ma, Y.~J. Rao, W.~L. Zhang, and B.~Hu,
  \enquote{\href{https://doi.org/10.1109/JSTQE.2018.2833472}{Multimode random
  fiber laser for speckle-free imaging},} {\protect\JournalTitle{IEEE Journal
  of Selected Topics in Quantum Electronics}} \textbf{25}, 1--6 (2018).

\bibitem{Azat_2022_Ghost_imaging}
A.~Ismagilov, A.~Lappo-Danilevskaya, Y.~Grachev, \emph{et~al.}, \enquote{\href{
  https://doi.org/10.1364/JOSAB.465222}{Ghost imaging via spectral multiplexing
  in the broadband terahertz range},} {\protect\JournalTitle{Journal of the
  Optical Society of America B}} \textbf{39}, 2335--2340 (2022).

\bibitem{Atsushi_2020_wood-plastic}
A.~Nakanishi and H.~Satozono, \enquote{\href{
  https://doi.org/10.1364/AO.379758}{Terahertz optical properties of
  wood–plastic composites},} {\protect\JournalTitle{Applied Optics}}
  \textbf{59}, 904--909 (2020).

\bibitem{Graham_2013_air-polymer}
G.~E. Town, S.~Ghatreh-Samani, S.~Busch, and M.~Koch, \enquote{\href{
  http://dx.doi.org/10.1109/IRMMW-THz.2013.6665697}{THz diffuser using an
  air-polymer composite material},} {\protect\JournalTitle{In: 2013 38th
  International Conference on Infrared, Millimeter, and Terahertz Waves
  (IRMMW-THz)}}  (2013).

\bibitem{Jaax_2013_Optical}
M.~Jaax, S.~Wolff, B.~Laegel, and H.~Fouckhardt, \enquote{\href{
  http://dx.doi.org/10.2971/jeos.2013.13020}{Optical and THz Galois
  diffusers},} {\protect\JournalTitle{Journal of the European Optical Society}}
  \textbf{8}, 13020 (2013).

\bibitem{naftaly_2007_terahertz}
M.~Naftaly and R.~E. Miles,
  \enquote{\href{https://doi.org/10.1109/JPROC.2007.898835}{Terahertz
  time-domain spectroscopy for material characterization},}
  {\protect\JournalTitle{Proceedings of the IEEE}} \textbf{95}, 1658--1665
  (2007).

\bibitem{GOODMAN_2007_Speckle}
J.~W. Goodman, \emph{\href{ https://doi.org/10.1007/s10955-007-9440-8}{Speckle
  phenomena in optics: theory and applications}} (Roberts and Company
  Publishers, 2007).

\bibitem{LI_2013_Coherence}
G.~Li, Y.~Qiu, and H.~Li, \enquote{\href{
  https://doi.org/10.1364/OE.21.013032}{Coherence theory of a laser beam
  passing through a moving diffuser},} {\protect\JournalTitle{Optical Express}}
  \textbf{21}, 13032--13039 (2013).

\bibitem{goodman_2005_introduction}
J.~W. Goodman, \emph{\href{https://doi.org/10.1063/1.3035549}{Introduction to
  Fourier optics}} (Roberts and Company publishers, 2005).

\bibitem{2004_Wang_SSIM_IQM}
Z.~Wang, A.~Bovik, H.~Sheikh, and E.~Simoncelli,
  \enquote{\href{https://doi.org/10.1109/TIP.2003.819861}{Image quality
  assessment: from error visibility to structural similarity},}
  {\protect\JournalTitle{IEEE Transactions on Image Processing}} \textbf{13},
  600--612 (2004).

\bibitem{2009_Thung_IQM}
K.-H. Thung and P.~Raveendran,
  \enquote{\href{https://doi.org/10.1109/TECHPOS.2009.5412098}{A survey of
  image quality measures},} in \emph{2009 International Conference for
  Technical Postgraduates (TECHPOS),}  (2009), pp. 1--4.

\bibitem{2010_Hore_IQM}
A.~Horé and D.~Ziou,
  \enquote{\href{https://doi.org/10.1109/ICPR.2010.579}{Image Quality Metrics:
  PSNR vs. SSIM},} in \emph{2010 20th International Conference on Pattern
  Recognition,}  (2010), pp. 2366--2369.

\bibitem{Mittal_2012_BRISQUE}
A.~Mittal, A.~K. Moorthy, and A.~C. Bovik,
  \enquote{\href{https://doi.org/10.1109/TIP.2012.2214050}{No-Reference Image
  Quality Assessment in the Spatial Domain},} {\protect\JournalTitle{IEEE
  Transactions on Image Processing}} \textbf{21}, 4695--4708 (2012).

\bibitem{Mittal_2013_NIQE}
A.~Mittal, R.~Soundararajan, and A.~C. Bovik, \enquote{Making a “completely
  blind” image quality analyzer,} {\protect\JournalTitle{IEEE Signal
  Processing Letters}} \textbf{20}, 209--212 (2013).

\bibitem{Hao_2021_Auto-focusing}
H.~Ding, F.~Li, Z.~Meng, \emph{et~al.},
  \enquote{\href{https://doi.org/10.1364/OE.434014}{Auto-focusing and
  quantitative phase imaging using deep learning for the incoherent
  illumination microscopy system},} {\protect\JournalTitle{Optics Express}}
  \textbf{29}, 26385--26403 (2021).

\bibitem{Yichen_2018_Extended}
Y.~Wu, Y.~Rivenson, Y.~Zhang, \emph{et~al.},
  \enquote{\href{https://doi.org/10.1364/OPTICA.5.000704}{Extended
  depth-of-field in holographic imaging using deep-learning-based autofocusing
  and phase recovery},} {\protect\JournalTitle{Optica}} \textbf{5}, 704--710
  (2018).

\bibitem{Yuan_2019_FourierImaging}
H.~Yuan, D.~Voß, A.~Lisauskas, \emph{et~al.},
  \enquote{{\href{https://doi.org/10.1063/1.5116553}{3D Fourier imaging based
  on 2D heterodyne detection at THz frequencies}},} {\protect\JournalTitle{APL
  Photonics}} \textbf{4}, 106108 (2019).

\bibitem{Xiang_2022_amplitudephase}
M.~Xiang, H.~Yuan, L.~Wang, \emph{et~al.},
  \enquote{\href{https://doi.org/10.48550/arXiv.2212.06725}{Amplitude/Phase
  Retrieval for Terahertz Holography with Supervised and Unsupervised
  Physics-Informed Deep Learning},}  (2022).

\end{thebibliography}

\end{document}